\documentclass[11pt]{article}
 
\usepackage{amssymb,amsmath,amsfonts}
\usepackage{graphicx,color,enumitem}
\usepackage{mathrsfs}
\usepackage{amsthm}
\usepackage[dvipsnames]{xcolor}
\usepackage{bm}
\usepackage{comment}
\usepackage{shuffle}
\usepackage[round]{natbib}
\usepackage[]{appendix}
\usepackage{xcolor}
\usepackage[a4paper, total={6in, 9in}]{geometry}
\usepackage{subcaption}
\usepackage{caption}
\usepackage{float}

\RequirePackage[colorlinks,citecolor=blue!70!green,urlcolor=blue, linkcolor=blue!70!green]{hyperref}

\usepackage{tikz}
\tikzstyle{vertex}=[circle, draw, inner sep=2pt, fill=white]

\renewcommand{\d}{{\mathrm{d}}}

\newcommand{\E}{{\mathbb E}}

\newcommand{\Q}{{\mathbb Q}}

\newcommand{\R}{{\mathbb R}}

\newcommand{\N}{{\mathbb N}}

\newcommand{\Kcal}{{\mathcal K}}

\newcommand{\Ucal}{{\mathcal U}}

\newcommand{\Ncal}{{\mathcal N}}

\newcommand{\Tcal}{{\mathcal T}}
\newcommand{\VIX}{{\textrm{VIX}}}
\newcommand{\SPX}{{\textrm{SPX}}}

\newcommand{\Vcal}{{\mathcal V}}

\newcommand{\y}{{y}}

\newcommand{\ind}{{\bf 1}}

\usepackage{listofitems} 
\usetikzlibrary{arrows.meta} 
\usepackage[outline]{contour}
\contourlength{1.4pt}
\tikzset{>=latex}
\usepackage{xcolor}
\colorlet{myred}{red!80!black}
\colorlet{myblue}{blue!80!black}
\colorlet{mygreen}{green!60!black}
\colorlet{myorange}{orange!70!red!60!black}
\colorlet{mydarkred}{red!30!black}
\colorlet{mydarkblue}{blue!40!black}
\colorlet{mydarkgreen}{green!30!black}
\tikzstyle{node}=[thick,circle,draw=myblue,minimum size=22,inner sep=0.5,outer sep=0.6]
\tikzstyle{node in}=[node,green!20!black,draw=mygreen!30!black,fill=mygreen!25] 
\tikzstyle{node hidden}=[node,blue!20!black,draw=myblue!30!black,fill=myblue!20]
\tikzstyle{node convol}=[node,orange!20!black,draw=myorange!30!black,fill=myorange!20]
\tikzstyle{node out}=[node,red!20!black,draw=myred!30!black,fill=myred!20]
\tikzstyle{connect}=[thick,mydarkblue] 
\tikzstyle{connect arrow}=[-{Latex[length=4,width=3.5]},thick,mydarkblue,shorten <=0.5,shorten >=1]
\tikzset{ 
  node 1/.style={node in},
  node 2/.style={node hidden},
  node 3/.style={node out},
}
\def\nstyle{int(\lay<\Nnodlen?min(2,\lay):3)} 

\usepackage{longtable}

\newtheorem{theorem}{Theorem}

\theoremstyle{definition}

\newtheorem{remark}[theorem]{Remark}

\newtheorem{lemma}[theorem]{Lemma}

\numberwithin{equation}{section}
\numberwithin{theorem}{section}

\newcommand{\cb}{\color{black}}

\usepackage{todonotes}

\begin{document}

\title{Pricing and calibration in the 4-factor path-dependent volatility model}
\author{Guido Gazzani\thanks{University of Verona, Department of Economics,
		Via Cantarane 24, 37129 Verona, Italy guido.gazzani@univr.it.} \and Julien Guyon\thanks{CERMICS, ENPC, Institut polytechnique de Paris, Marne-la-Vallée, France, julien.guyon@enpc.fr.}\;\,\thanks{NYU Tandon School of Engineering, Department of Finance and Risk Engineering, One MetroTech Center, Brooklyn, NY 11201, USA, julien.guyon@nyu.edu. \newline
The authors acknowledge financial support from the  BNP Paribas Chair \emph{Futures of Quantitative Finance}. The first author acknowledges financial support by the “PHC AMADEUS” program (project number: 47561RJ), funded by the French Ministry for Europe and Foreign Affairs, the French Ministry for Higher Education and Research, and the Austrian Ministry for Higher Education. The present work was initiated when the first author was affiliated to CERMICS, ENPC.}}

\maketitle
\begin{abstract}
  We consider the path-dependent volatility (PDV) model of Guyon and Lekeufack (2023), where the instantaneous volatility is a
    linear combination of a weighted sum of past returns and the square root of a weighted sum of past squared returns. We discuss the influence
    of an additional parameter that unlocks enough volatility on the upside to reproduce the implied volatility smiles of S\&P 500 and VIX options. This PDV model, motivated by empirical studies, comes with computational challenges, especially in relation to VIX options pricing and calibration. We propose an accurate {\cb \emph{pathwise}} neural network approximation of the VIX which leverages on the Markovianity of the 4-factor version of the model. The VIX is learned {\cb pathwise} as a function of the Markovian factors and the model parameters. We use this approximation to tackle the joint calibration of S\&P 500 and VIX options, {\cb quickly sample VIX paths, and price derivatives that jointly depend on S\&P 500 and VIX. As an interesting aside, we also show that this \emph{time-homogeneous}, low-parametric, Markovian PDV model is able to fit the whole surface of S\&P 500 implied volatilities remarkably well.}
\end{abstract}

\section{Introduction}

\subsection{Path-dependent volatility models}

In this article, we are interested in calibrating the path-dependent volatility (PDV) model of \cite{GLJ:22} to S\&P 500 (SPX) and VIX futures and options, {\cb in quickly sampling VIX paths, and in pricing derivatives that may jointly depend on S\&P 500 and VIX}. PDV models aim to provide an explicit and accurate description of the purely endogenous part of the joint dynamics of an asset price and its volatility (``spot-vol dynamics''): the volatility is described as a deterministic, explicit function of past asset returns. This feedback effect from past returns to volatility, which in turn impacts future returns, creates very rich, intricate, nonlinear dynamics of asset prices.

PDV models have recently received increased attention for the following reasons:
\begin{itemize}
    \item Empirical data confirms the path-dependent nature of volatility, see for instance \cite{Z:09,Z:10,CB:14,BDB:17,GLJ:22,ABJ:23}, as well as the ARCH-GARCH literature.
    \item PDV models can naturally reproduce all the following important stylized facts about volatility: leverage effect, volatility clustering, heavy distribution tails, roughness at the daily scale, weak and strong Zumbach effects, jump-like behavior with very fast large volatility spikes followed by slower decays.
\end{itemize}
Like local volatility models, PDV models are complete models, but they can generate much more realistic joint spot-vol dynamics, as they capture the path-dependent nature of volatility. Of course, volatility is not purely path-dependent: some unanticipated, exogenous news have an impact on the volatility of financial markets. Stochastic volatility (SV) models start from the exogenous part, postulate a dynamics for the instantaneous volatility which typically depends on unobservable factors, such as Ornstein-Uhlenbeck factors, and can only generate some implicit, complicated path-dependency by correlating the Brownian motions that drive the dynamics of the asset price with those that drive the dynamics of the SV. Of course, one can build a fully PDV model from an SV model by fully correlating those Brownian motions; see for instance the quadratic rough Heston model by \cite{GJR:20} and the path-dependent two-factor Bergomi model in \cite{G:22} where it is shown that the joint SPX/VIX smile calibration of the two-factor Bergomi model degenerates the model into its path-dependent version. However, PDV  is not about fully correlating the Brownian motions of an underlying SV model. The PDV paradigm is different. It suggests another, more natural approach to volatility modeling:
\begin{enumerate}
    \item First, model the (large) purely endogenous part of volatility \emph{explicitly} as well as possible, using PDV models that depend on \emph{observable} factors, such as past asset returns.
    \item Then, model the (smaller) exogenous part, for example by analyzing empirical residuals of the (pure) PDV model. 
\end{enumerate}
We call the resulting model a \emph{path-dependent stochastic volatility} (PDSV) model.
The class of explicit PDV models  includes for instance: many contributions in the ARCH-GARCH literature, e.g., \cite{S:95}, the complete model of \cite{HR:98} also discussed by \cite{FP:05}, the diagonal QARCH model of \cite{CB:14}, the ZHawkes process of \cite{BDB:17}, the PDV model of \cite{GLJ:22}, as well as some version of the EWMA Heston of \cite{P:23}. In particular, the PDV model of \cite{GLJ:22} was shown to have a better predictive power of both implied volatility and future realized volatility, across equity indexes and train/test periods, than the best competing models available in the literature.

\medskip

In the {continuous-time version of the} PDV model of \cite{GLJ:22}, the instantaneous volatility $\sigma_t$ takes a very simple form: {it is an affine combination of} a weighted sum $R_1$ of past returns, and the square root of a weighted sum $R_2$ of past squared returns:
\begin{equation*}
    \sigma_t = \beta_0 + \beta_1 R_{1,t} + \beta_2 \sqrt{R_{2,t}}.
\end{equation*}
This particular, simple specification of a PDV came as the result of a {thorough} empirical study where more complicated features and functional forms ({in particular,} neural networks, including LSTM networks) were tested.

One distinctive feature of this PDV model is that it is homogeneous in volatility, in particular capturing the dependency on past realized volatility through the term $\sqrt{R_2}$. The other PDV models that have been studied in the literature are either homogeneous in variance, or mix terms homogeneous to a volatility with terms homogeneous to a variance, and/or {lack} the important $R_2$ feature.

This distinctive feature may explain the higher predictive power of the model. However, it also means that the model comes with computational challenges: unlike other models which are chosen because of their mathematical tractability, such as affine or polynomial models, it does not seem to have a nice mathematical structure allowing for fast pricing. In particular, in this PDV model, the forward instantaneous variances and the VIX are not directly accessible, since the instantaneous variance $\sigma_t^2$ contains $R_1\sqrt{R_2}$ and $\sqrt{R_2}$ terms whose conditional expectations cannot be exactly computed.

{\cb The main objective of this article is to address those computational challenges and compute the VIX in the model. We propose an accurate \emph{pathwise} neural network approximation of the VIX which leverages on the Markovianity of the 4-factor version of the model. The VIX is learned pathwise as a function of the Markovian factors and the model parameters. This approximation is then used to quickly simulate VIX paths, jointly calibrate the model to SPX and VIX options, and quickly price derivatives that jointly depend on the VIX.} 

\subsection{The joint calibration problem}

Jointly calibrating to SPX and VIX futures and options is important to prevent arbitrage and ensure accurate pricing of liquid hedging instruments. {\cb Indeed, while calibrating to SPX options means incorporating market information on SPX spot implied volatilities, calibrating to VIX derivatives means incorporating market information on \emph{future} SPX implied volatilities}, an information that is not contained in SPX option prices. For instance, \cite{GB:22} show the additional information contributed by VIX futures and VIX smiles, by quantifying how model-free bounds for SPX path-dependent payoffs tighten when the prices of VIX futures and VIX options are included. They also show that, to avoid mispricing some payoffs (particularly forward-starting payoffs, which are most sensitive to forward volatilities), it is important that the model not only fits SPX options but also fits VIX options, even when the payoff depends only on SPX prices.

In the rest of this section we review some of the literature on the joint calibration problem. We refer the reader to \cite{DKMY:23} for a recent survey on volatility modeling and advances on the joint calibration problem.
We split the literature in two main streams of research: exact (nonparametric) and approximate (parametric) joint fits.

The first exact joint fit was obtained in \cite{G:20a}, where a jointly calibrated {nonparametric} discrete-time model was built by minimum entropy, i.e., by solving a Schr\"odinger  problem; {see also \cite{G:24}}. Efficient algorithms for this problem (in particular, faster than the Sinkhorn algorithm) have recently been presented in \cite{GB:22}, who also extended the model to continuous time by martingale interpolation. A direct continuous-time Schr\"odinger bridge approach, inspired by \cite{HL:19}, is presented in \cite{G:22a}. A similar direct nonlinear optimal transport approach was suggested in \cite{GLOW:21}. The direct continuous-time approaches are much more computationally demanding than the novel discrete-time-continuous-time procedure of \cite{GB:22}. Note that an exact joint calibration ensures the absence of joint SPX/VIX arbitrages.

The other line of research concerns approximate fits of parametric models. The PDV model of \cite{GLJ:22} falls into this category. The benefit of low-parametric models is their interpretability. However, they cannot fit market smiles exactly, since the low-parametric structure induces rigidities. The main objective { of a modeler} is therefore to design parametric models that are able to accurately fit market SPX and VIX smiles. {As explained in \cite{G:20}}, the main challenge is to fit at the same time large short-term SPX at-the-money (ATM) skews with relatively low VIX implied volatilities. In a first attempt, \cite{G:08} used a double constant elasticity of variance model, which, despite {its quite large number of parameters}, cannot jointly fit the implied volatilities of SPX and VIX options accurately enough for trading purposes. Later on, several authors have included jumps in the dynamics of the asset price and/or its volatility, e.g., the forward variance model of \cite{CK:13} described as the exponential of an affine process with L\'evy jumps which allows for Fourier pricing, a regime-switching enhancement of the classical Heston model by \cite{PS:14}, the 3/2 model with jumps in the asset price of \cite{BB:14}, in the volatility (\cite{KS:15}), or with co-jumps and idiosyncratic jumps in the volatility (\cite{PPR:18}).

Continuous-paths models have also been employed to try to solve the joint calibration problem. For instance, in \cite{FS:18}, a Heston model with stochastic vol-of-vol is calibrated, but only for maturities above 4 months when VIX options are less liquid. More recently, \cite{R:22} considered a model where the volatility is driven by two Ornstein-Uhlenbeck (OU) processes using a hyperbolic transformation function. Unfortunately, the study of \cite{R:22} does not address the calibration to VIX futures, which is crucial as VIX futures are the tradable assets on which VIX options are written. \cite{GLJ:22} show that the 4-factor PDV (4FPDV) model, a low-parametric Markovian model with only 10 parameters, is able to jointly fit SPX smiles, VIX futures, and VIX smiles with good accuracy. Some parametric models use a much large number of parameters, such as signature-based models or neural SDE models. The use of two OU processes inspired the first contribution with signature-based models, see \cite{CGS:22,CGMS:23}. These models address the joint calibration problem by training only a subset of model parameters, namely the linear read out map, as is usually done in reservoir computing. The class of Gaussian polynomial processes introduced in \cite{AJIL:22}, which can be seen as a special case of signature-based models, lends itself to a quantization pricing technique for SPX and VIX options. Using neural SDEs, \cite{GM:22} show that the joint calibration problem can be very accurately solved by a one-factor stochastic local volatility model, provided enough flexibility is allowed in the dynamics.

Rough volatility models, e.g., the quadratic rough Heston model (\cite{GJR:20,RZ:21}) have also been used to tackle the joint calibration problem with moderate success, due to their very parsimonious, low-parametric nature. \cite{R:22} indeed shows that the rough Bergomi model, the rough Heston model, and an extended rough Bergomi model are all outperformed in the joint calibration task by a hyperbolic 2-OU model; see also \cite{AJS:24} {where different Bergomi-type models in this regard have been compared in terms of their calibration performances on SPX options only}. {
The rough Heston model 
belongs to the class of affine Volterra processes considered in \cite{ALP:17, CT:19}, and thus allows for Fourier pricing after solving the  corresponding Riccati equations. This underlying structure is the building block of an extension with jumps (see \cite{BLP:22}) that is employed in the context of the joint calibration by \cite{BPS:22}, where a rough Heston model with Hawkes-type jumps is shown, using Fourier pricing, to solve the joint calibration problem for short maturities with good accuracy.

\subsection{Main contributions}

The main contributions of this article are the following.

\begin{itemize}
    \item The joint calibration procedure in \cite{GLJ:22} was manual, and the optimization depended strongly on an initial guess. In this work we present an automatic, more robust procedure.
    \item Moreover, the calibration in \cite{GLJ:22} was very slow, since the VIX was computed by nested Monte Carlo. In order to build a faster calibration procedure, we provide a novel {\cb pathwise} deep learning approach to simulating the VIX under a given parametric Markov model, in our case the 4FPDV model, to efficiently price VIX futures and VIX options.
    \item With the addition of a single parameter we generate enough flexibility in the 4FPDV model to solve the joint SPX/VIX calibration problem for short maturities, leveraging on the neural approximation of the VIX. 
    \item We also show that the {4FPDV} model calibrates the SPX implied volatility surface very well, over a large range of maturities (up to at least 1 year).
    \item We provide an analysis of the stability over time of the {calibrated} parameters.
    \item {\cb Finally, we use our pathwise approximation of the VIX to quickly price light exotics whose payoffs depend both on SPX and VIX in the 4FPDV model.}
\end{itemize}

The remainder of this article is structured as follows. Section \ref{sec:model} is a reminder about the 4FPDV model. The novel {\cb pathwise} deep learning estimation of the VIX is presented in Section \ref{sec:NN}. Calibration results are presented in Section \ref{sec:calib_SPX} for the calibration to the SPX implied volatility surface only and in Section \ref{sec:joint} for the joint SPX/VIX calibration. The stability analysis is carried out in Section \ref{sec:stability}. {\cb Finally, pricing examples are reported in Section \ref{sec:pricing_light_exotics}.}
The market data used in our numerical experiments was provided by BNP Paribas via the Chair \emph{Futures of Quantitative Finance.}
{\cb The codes for the present work are available at \href{https://github.com/GuidoGazzani-ai/pdv_nn}{GuidoGazzani-ai/pdv$\_$nn}.}

\section{The model}\label{sec:model}

Throughout this article, $(r_{t})_{t\ge0}$ and $(q_{t})_{t\ge0}$ denote the interest rate and the dividend yield curves, respectively, and are assumed deterministic. The dividend yield includes the repo rate.  We recall the 4-factor path-dependent volatility (4FPDV) model of \cite{GLJ:22}.  The price $(S_{t})_{t\geq0}$ of the asset is modeled under the unique risk-neutral probability measure $\Q$ by
\begin{equation}\label{model}
    \d S_t = (r_t-q_t)S_t \,\d t + S_t \sigma_t \,\d W_t,
\end{equation}
where $S_{0}>0$, $\sigma=(\sigma_{t})_{t\geq0}$ is the instantaneous volatility process to be specified, and $W=(W_{t})_{t\geq0}$ is a standard one-dimensional $\Q$-Brownian motion. 
The volatility is path-dependent and takes the following simple explicit form as a function of past returns:
\begin{align}\label{eq:vol_dynamics} \nonumber
    \sigma_{t}&:=\sigma(R_{1,t},R_{2,t}),\\ 
    \sigma(R_{1},R_{2})&:=\beta_{0}+\beta_{1}R_{1}+\beta_{2}\sqrt{R_{2}} +\beta_{1,2}R_{1}^{2}\ind_{\{R_{1}>0\}},\\ \nonumber
    R_{1,t}&:=(1-\theta_{1})R_{1,0,t}+\theta_{1}R_{1,1,t},\\ \nonumber
    R_{2,t}&:=(1-\theta_{2})R_{2,0,t}+\theta_{2}R_{2,1,t},\nonumber
\end{align}
where $R_{1,0,t}$, $R_{1,1,t}$, $R_{2,0,t}$, $R_{2,1,t}$ are four path-dependent factors
defined by\footnote{ By $(\frac{\mathrm{d}S_{u}}{S_{u}})^{2}$ we mean $\sigma_u^2 \,\d u$.}
\begin{equation*}
    R_{n,p,t}=\int_{-\infty}^{t} \lambda_{n,p}e^{-\lambda_{n,p}(t-u)}\left(\frac{\mathrm{d}S_{u}}{S_{u}}\right)^{n},\qquad n\in\{1,2\}, \quad p\in\{0,1\}.
\end{equation*}
The two factors $R_{1,p,t}$ are exponentially weighted moving averages of past returns, measuring the recent trend in the asset price, while the two factors $R_{2,p,t}$ are exponentially weighted moving averages of past squared returns, measuring the recent volatility in the asset price, with exponential weights $\lambda_{n,p}>0$. Note that
\begin{equation*}
R_{n,t}=\int_{-\infty}^{t} K_{n}(t-u)\left(\frac{\mathrm{d}S_{u}}{S_{u}}\right)^{n}, \qquad n\in\{1,2\},
\end{equation*}
is a weighted average of past returns (when $n=1$) or past squared returns (when $n=2$), where the convolution kernels
\begin{equation}\label{eq:kernel}
    K_{n}(t):=(1-\theta_{n})\lambda_{n,0}e^{-\lambda_{n,0}t}+\theta_{n}\lambda_{n,1}e^{-\lambda_{n,1}t}
\end{equation}
are convex combinations of two exponential kernels. {\cb We use the convention that $\lambda_{n,0}\ge \lambda_{n,1}$, only for interpretation purposes, so $\lambda_{n,0}$ codes for the short memory and $\lambda_{n,1}$  for the long memory.} This parametrization of the kernels allows us to
\begin{enumerate}
    \item mix short memory (large $\lambda_{n,0}$) and long memory (smaller $\lambda_{n,1}$)---in particular, such kernels are well approximated by power laws over a large range of maturities;
    \item  build a Markovian model, since
    \begin{align*}
        \mathrm{d}R_{1,p,t}&=-\lambda_{1,p}R_{1,p,t}\,\d t+\lambda_{1,p}\sigma(R_{1,t},R_{2,t})\,\d W_{t},\\ \nonumber
        \mathrm{d}R_{2,p,t}&=\lambda_{2,p}(\sigma(R_{1,t},R_{2,t})^2-R_{2,p,t})\,\d t. \nonumber
    \end{align*}
\end{enumerate}

{\cb \cite{NV:23} have established the strong existence and uniqueness of a solution $(R_{1,0,t},R_{1,1,t},R_{2,0,t},R_{2,1,t})$ of the above SDE, hence the wellposedness of the 4FPDV model, when $\beta_{1,2}=0$. \cite{AJ:24} have extended this result to the case of generic kernels. In particular, the instantaneous volatility cannot blow up in finite time.}

The parameters have the following interpretation:
\begin{itemize}
    \item $\lambda_{1,0}$ captures the dependence of $R_1$ on recent returns; the larger $\lambda_{1,0}$, the more weight is given to recent returns;
    \item $\lambda_{1,1}<\lambda_{1,0}$ captures the dependence of $R_{1}$ on older returns; the smaller $\lambda_{1,1}$, the more weight is given to old returns;
    \item $\lambda_{2,0}$ captures the dependence of $R_2$ on recent squared returns; the larger $\lambda_{2,0}$, the more weight is given to recent squared returns;
    \item $\lambda_{2,1}<\lambda_{2,0}$ captures the dependence of $R_{2}$ on older squared returns; the smaller $\lambda_{2,1}$, the more weight is given to old squared returns;
    \item $\theta_{1}\in[0,1]$ (resp., $\theta_2\in[0,1]$) mixes the dependence on recent and older returns (resp., squared returns) to form the summary variable $R_1$ (resp., $R_2$);
    \item $\beta_0>0$ is the baseline instantaneous volatility;
    \item $\beta_{1}<0$ and $\beta_{2}>0$ are the sensitivities of the volatility to the trend $R_1$ and the historical volatility $\sqrt{R_2}$, respectively;
    \item $\beta_{1,2}\ge 0$ is an additional parameter which produces volatility when the asset price trends up---see below.
\end{itemize}

\paragraph{The role of $\beta_{1,2}$.}
The original 4FPDV model, which aimed {to explain} a single global level of implied volatility (e.g., the VIX) using past daily returns, has $\beta_{1,2}=0$. As explained in \cite{GLJ:22}, allowing $\beta_{1,2}$ to vary did not significantly increase the predictive power of the model. However, the continuous-time limit of the original 4FPDV model, that is, Model \eqref{eq:vol_dynamics} with $\beta_{1,2}=0$, generates out-the-money (OTM) call implied volatilities that are too small and do not increase with strike for large strikes. In order for the continuous-time model to fit not only one single global level of implied volatility but full smiles of implied volatility, \cite{GLJ:22} added the term
$\beta_{1,2}R_{1}^{2}\ind_{\{R_{1}>0\}}$, with $\beta_{1,2}\ge 0$. Indeed, compared to the case where $\beta_{1,2}=0$, this term generates added volatility when the asset price experiences a positive trend. This allows the asset price to take large values with higher probability, and this results in increased risk-neutral density for large asset price values, therefore increased implied volatility for large strikes.
\begin{table}[H]
    \begin{center} 
    \caption{}
\label{tab:parabolic}
\begin{tabular}{||c |c|c| c|c|c||} 
 \hline
 $\lambda_{1,0}=62.11$ & $\lambda_{1,1}=32.25$& $\theta_1=0.23$  & $\lambda_{2,0}=9.57$ & $\lambda_{2,1}=3.51$  & $\theta_2=0.99$\\ 
 \hline
\end{tabular}
\begin{tabular}{||c |c|c ||} 
 \hline
 $\beta_{0}=0.026$ & $\beta_{1}=-0.138$ & $\beta_{2}=0.69$ \\ 
 \hline
\end{tabular}
\end{center}
\end{table}

\begin{figure}
     \centering
         \includegraphics[trim={0.7cm 0.7cm 0.7cm 0.7cm},clip,scale=0.37]{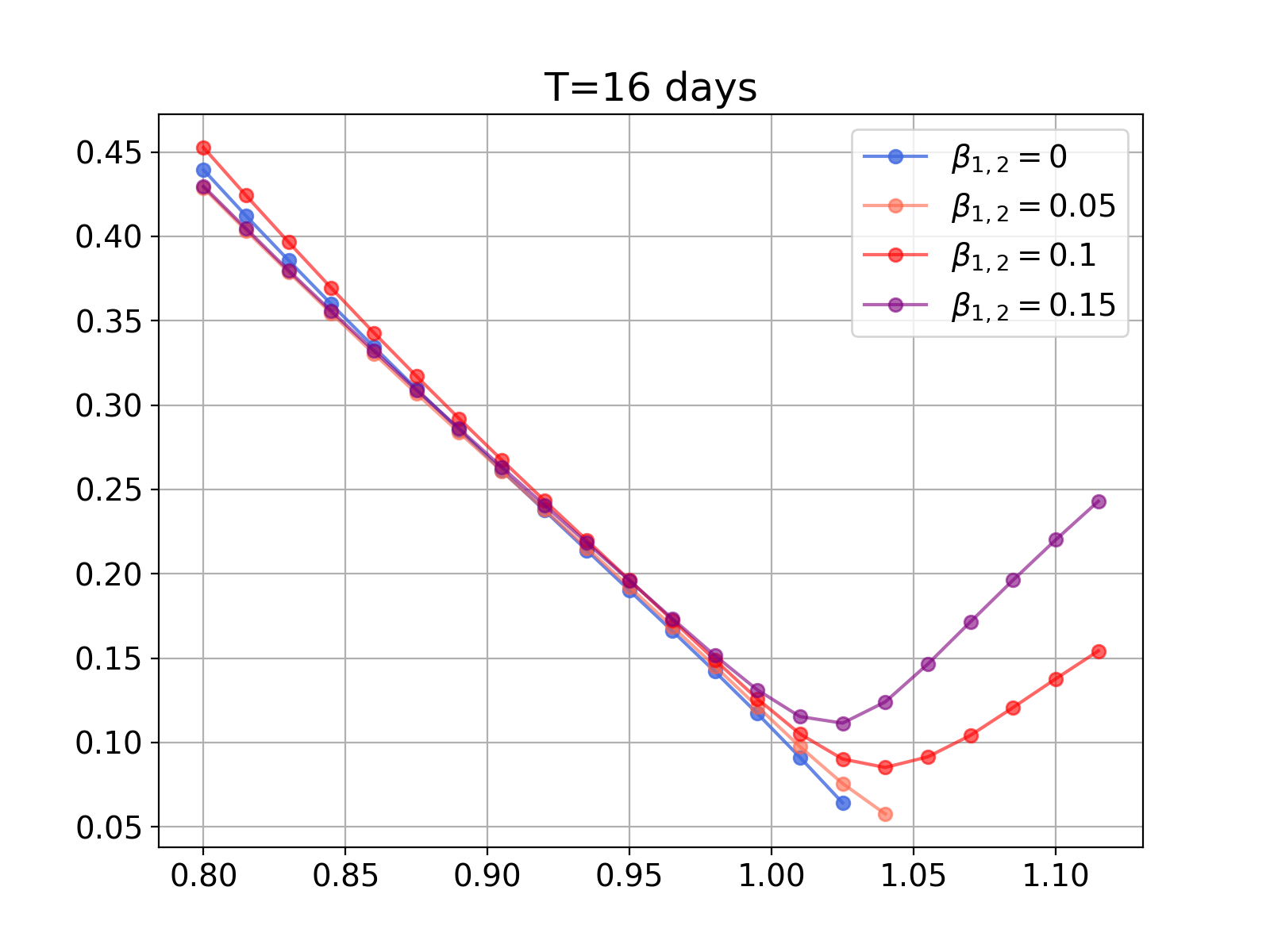}
         \includegraphics[trim={0.7cm 0.7cm 0.7cm 0.7cm},clip,scale=0.37]{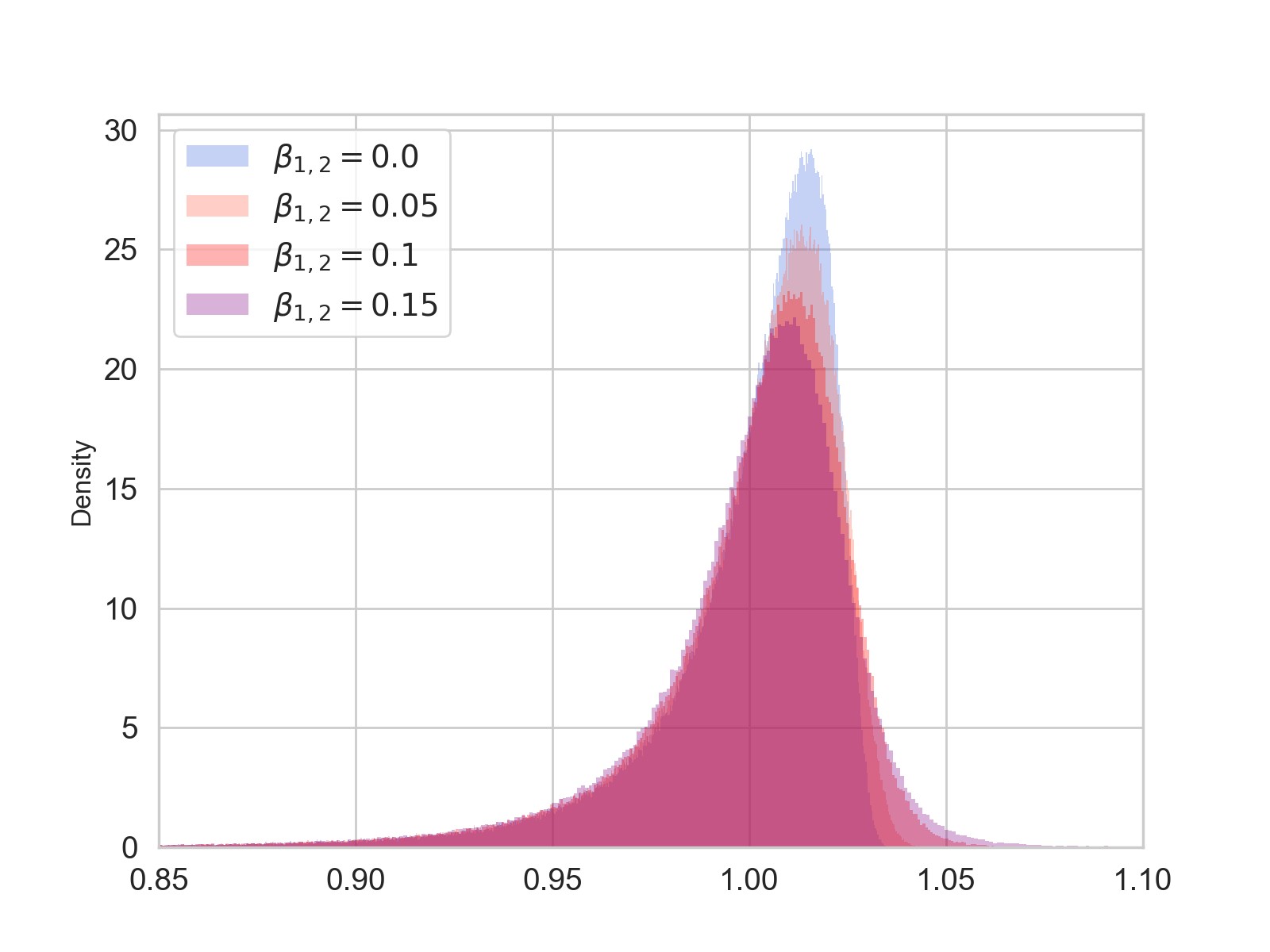}\\
        \includegraphics[trim={0.7cm 0.4cm 0.7cm 0.4cm},clip,scale=0.37]{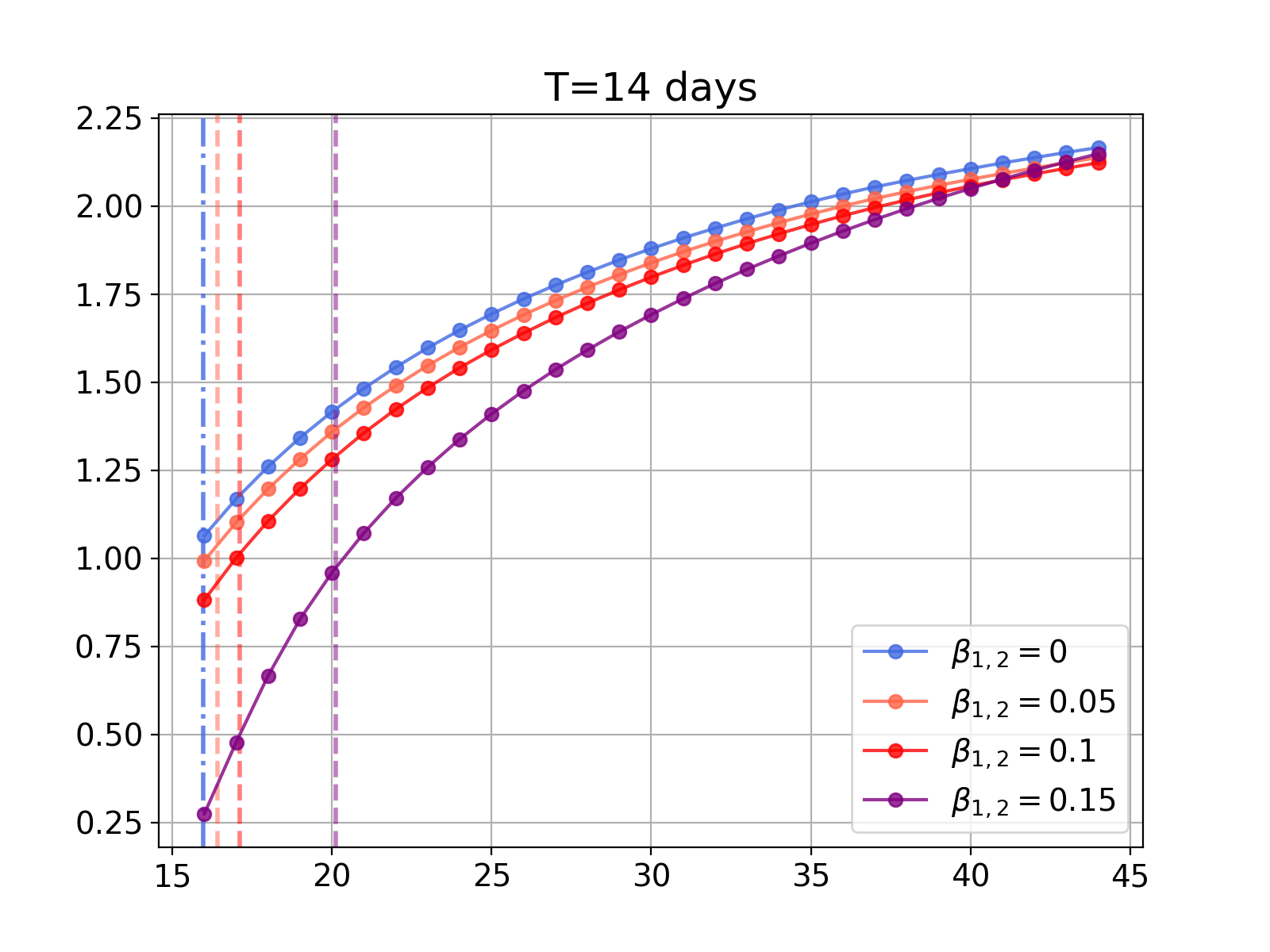}
        \caption{Impact of $\beta_{1,2}$ on model SPX implied volatilities at maturity 16 days (top left), the corresponding model SPX density (top right), and model VIX futures and VIX implied volatilities (bottom). Model parameters are reported in Table \ref{tab:parabolic}, except $\beta_{1,2}$ which varies in $\{0,0.05,0.1,0.15\}$. Prices of call options on $S$ are computed by Monte Carlo 
        using $N_{\text{MC}}=10^{5}$ trajectories and a discretization step $\Delta t=\frac{1}{2520}$. VIX call option prices are computed using $3\cdot 10^3$ nested Monte Carlo paths. Initial values of the factors: $R_{1,0,0}=0.2988$, $R_{1,1,0}=0.2397$, $R_{2,0,0}=0.016$, $R_{2,1,0}=0.02$.}
         \label{fig:impact_beta12}
\end{figure}

This is illustrated in Figure \ref{fig:impact_beta12} (top), where we plot the model SPX smile and risk-neutral distribution at maturity 16 days. We use the model parameters from Table 8 in \cite{GLJ:22}, calibrated to market smiles, except that we vary $\beta_{1,2}\in\{0,0.05,0.1,0.15\}$; the parameters are reported in Table \ref{tab:parabolic}. Figure \ref{fig:impact_beta12} (top left) shows that the extra term $\beta_{1,2}R_{1}^{2}\ind_{\{R_{1}>0\}}$ allows us to reproduce the hockey-stick shape  that is typical of market SPX smiles; Figure \ref{fig:impact_beta12} (top right) illustrates the increased risk-neutral density on the large SPX side.  Figure \ref{fig:impact_beta12} (bottom) shows the
model VIX smile with maturity 14 days. The increased $\VIX$ future value, denoted by a dashed line, results from the additional volatility.

\bigskip 

The dynamics of the instantaneous volatility in model \eqref{eq:vol_dynamics} is given by the following lemma. For ease of notation, we denote $R_t:=(R_{1,0,t},R_{1,1,t},R_{2,0,t},R_{2,1,t})$ and introduce the following quantities:
\begin{equation*}
    \bar{\lambda}_{n}:=(1-\theta_{n})\lambda_{n,0}+\theta_{n}\lambda_{n,1},\quad \bar{R}_{n,t}:=\frac{(1-\theta_{n})\lambda_{n,0}R_{n,0,t}+\theta_{n}\lambda_{n,1}R_{n,1,t}}{\bar{\lambda}_n}, \quad n\in\{1,2\}.
\end{equation*}

\begin{lemma}
    Let $\sigma=(\sigma_{t})_{t\ge0}$ satisfy \eqref{eq:vol_dynamics}. Then
    \begin{equation*}
        \mathrm{d}\sigma_{t}=\mu(R_t)\,\d t+\nu(R_t)\,\d W_t,
    \end{equation*}
    where
    \begin{align*}
        \mu(R_t)&=-(\beta_{1}+2\beta_{1,2}R_{1,t}\ind_{\{R_{1,t}>0\}})\bar{\lambda}_1\bar{R}_{1,t}+\beta_2\bar{\lambda}_{2}\frac{(\sigma_{t}^2-\bar{R}_{2,t})}{2\sqrt{R_{2,t}}} +\beta_{1,2}\bar{\lambda}_{1}^2 \sigma_t^2,\\
        \nu(R_t)&=(\beta_{1}+2\beta_{1,2}R_{1,t}\ind_{\{R_{1,t}>0\}})\bar{\lambda}_{1}\sigma_{t}.
    \end{align*}
\end{lemma}
\begin{proof}
    The proof is a straightforward application of  It\^o's formula.
\end{proof}

When $\beta_{1,2}=0$, the instantaneous (lognormal) volatility of the instantaneous volatility is constant (equal to $|\beta_{1}|\bar{\lambda}_{1}$), like in Bergomi models.
When $\beta_{1,2}>0$, it is equal to $|(\beta_{1}+2\beta_{1,2}R_{1,t}\ind_{\{R_{1,t}>0\}})|\bar{\lambda}_{1}$ and thus depends on the trend $R_{1,t}$. It vanishes when $R_{1,t} = -\frac{\beta_1}{2\beta_{1,2}}>0$, and can get large when the asset price trends very positively. {\cb In the numerical experiments, to avoid a possible explosion of the instantaneous volatility caused by the superlinear term $R_1^2$, the instantaneous volatility $\sigma$ is capped at 1.5.}

{\cb \begin{remark}
    Let us stress that in classical stochastic volatility models, the initial values of the instantaneous volatility or of the latent volatility factors, e.g., Ornstein-Uhlenbeck factors, are not observable and thus have to be calibrated, like model parameters. The same holds for Markovian approximations of rough volatility models, e.g., \cite{RZ:21}. By contrast, in pure PDV models, the factors driving the instantaneous volatility are \emph{observable} in the market---in the case of the 4FPDV model, weighted averages of past returns and past squared returns. The model parameters thus directly yield the initial values of the factors, which do not need to be calibrated.
\end{remark}}

\section{Learning the VIX with neural networks}\label{sec:NN}

{\cb We aim to fastly compute the VIX in the 4FPDV model, pathwise, at any date in the future. This will allow us to quickly price derivatives depending on the VIX, generate VIX paths, and build a computer procedure to jointly calibrate the 4FPDV model to SPX smiles, VIX futures, and VIX smiles.} In view of calibrating the model to VIX futures and  VIX smiles, we will not only learn VIX$_T$ as a function of the Markovian factors $(R_{1,0,T},R_{1,1,T},R_{2,0,T},R_{2,1,T})$, but also as a function of the model parameters.

\subsection{A brief reminder on the VIX}

The CBOE Volatility Index, also known as the VIX, is a popular measure of the market's expected volatility of the SPX index over the next 30 days. It is calculated and published by the Chicago Board Options Exchange (CBOE). The VIX is meant to represent the 30-day implied volatility of the log-contract on the SPX index.
We thus take as stylized definition of the VIX
\begin{equation}\label{initial_formula_vix_mkt}
    \VIX_{T}^2:=\text{Price}_T\left\lbrack -\frac{2}{\Delta}\log\left(\frac{S_{T+\Delta}}{F_{T}^{T+\Delta}}\right)\right\rbrack,
\end{equation}
where $\Delta=30$ days and $F_{t}^{T+\Delta}$ denotes the price at $t$ of the SPX future with maturity $T+\Delta$. 
Within a pricing model, we therefore define the VIX by
\begin{equation}\label{initial_formula_vix_model}
    \VIX_{T}^2:=\mathbb{E}\left\lbrack -\frac{2}{\Delta}\log\left(\frac{S_{T+\Delta}}{F_{T}^{T+\Delta}}\right)\middle| \mathcal{F}_{T}\right\rbrack,
\end{equation}
where $\E[\cdot]$ denotes the expectation under the pricing measure $\Q$ and $(\mathcal{F}_{t})$ denotes the filtration representing the flow of available information in the market. Note that the nonnegativity of the r.h.s.\ of \eqref{initial_formula_vix_mkt} results from absence of arbitrage, while that of the r.h.s.\ of \eqref{initial_formula_vix_model} results from the martingale property of $(F_t^{T+\Delta})_{0\le t\le T+\Delta}$, the concavity of the logarithm, and Jensen's inequality.
 
In order to compute the VIX in model \eqref{eq:vol_dynamics}, we thus need to compute the 
conditional expectation \eqref{initial_formula_vix_model}, where $(\mathcal{F}_{t})$ is the filtration generated by the Brownian motion $W$. From \cite{N:94,D:94}, since the SPX model price has continuous paths,
\begin{equation}\label{initial_formula_vix2}
    \VIX_{T}^2=\mathbb{E}\left\lbrack \frac{1}{\Delta}\int_{T}^{T+\Delta} \sigma_t^2 \,\d t\middle| \mathcal{F}_{T}\right \rbrack.
\end{equation}
This equality, which immediately follows from \eqref{initial_formula_vix_model} by applying the It\^o formula to $\log(F_t^{T+\Delta})$ between $T$ and $T+\Delta$, explains the choice \eqref{initial_formula_vix_mkt} to represent the expected volatility of the SPX index over the next 30 days.
We observe that when $\sigma_{t}$ satisfies \eqref{eq:vol_dynamics}, the expression \eqref{initial_formula_vix2} is not known in closed form, i.e., one cannot explicitly compute 
\begin{equation}\label{initial_formula_vix_model_PDV}
    \VIX_{T}^2=\frac{1}{\Delta}\mathbb{E}\left\lbrack\int_{T}^{T+\Delta}\big(\beta_{0}+\beta_{1}R_{1,t}+\beta_{2}\sqrt{R_{2,t}} +\beta_{1,2}R_{1,t}^{2}\ind_{\{R_{1,t}>0\}}\big)^2\,\d t \middle| \mathcal{F}_{T}\right \rbrack.
\end{equation}
Nevertheless, due to the Markovianity of $R_{t}=(R_{1,0,t},R_{1,1,t},R_{2,0,t},R_{2,1,t})$ and the time-homogeneity of its dynamics, the VIX squared at time $T>0$ is a measurable function $f(\Theta,R_{T})$, where $\Theta$ denotes the collection of the model parameters,
\begin{equation*}
    \Theta=(\lambda_{1,0},\lambda_{1,1},\theta_{1},\lambda_{2,0},\lambda_{2,1},\theta_{2},\beta_{0},\beta_{1},\beta_{2},\beta_{1,2})\in\mathbb{R}^{10}.
\end{equation*}
Our key idea is to learn the function $f$ by parameterizing it as a feed-forward neural network.

{\cb
\begin{remark}\label{rem:nn_approaches}
An important remark is that we use neural networks to directly learn the VIX at future dates, \emph{pathwise}, as a function of the model parameters and the Markovian factors---not just the prices at time 0 of VIX futures and VIX options as a function of the model parameters. Our use of neural networks thus differs from the so-called ``deep pricing" and ``deep calibration" approaches, which we briefly recall. Let $\mathcal{T}$ and $\mathcal{K}$ be collections of maturities and strikes, respectively. Using the notation $\phi_{T,K}$ to represent the price of a European vanilla option or its implied volatility with maturity $T\in\mathcal{T}$ and strike $K\in\mathcal{K}$, the ``deep pricing" and ``deep calibration" methodologies can be summarized in the three following network parametrizations:
\begin{itemize}
    \item Global deep pricing: $\Ncal\Ncal: \Theta\mapsto (\phi_{T,K}^{\text{model}})_{T\in\mathcal{T},K\in\mathcal{K}}$ as in \cite{HMT:21,RZ:21,R:22}, for fixed $\mathcal{T}$ and $\mathcal{K}$.
    \item Pointwise deep pricing: $\Ncal\Ncal: (\Theta,T,K)\mapsto \phi_{T,K}^{\text{model}}$ as in \cite{BS:18,BBR:23}.
    \item Deep calibration: $\Ncal\Ncal: (\phi_{T,K}^{\text{model}})_{T\in\mathcal{T},K\in\mathcal{K}}\mapsto \Theta$ as for instance in \cite{H:16,GT:20} to directly solve the calibration problem.
\end{itemize}
Those methodologies are most relevant for non-Markovian models, for which Monte Carlo simulations can be very costly. In this article, we learn pathwise quantities (here, the VIX), i.e., we learn at the finest granular level, and our calibration procedure will involve Monte Carlo sampling. One could directly use neural networks to learn the prices at time 0 of VIX futures and VIX options, like in the ``deep pricing" and ``deep calibration" approaches, without learning the VIX pathwise. However, learning the VIX pathwise as a function of the factors and model parameters is necessary to (1) quickly price derivatives involving the VIX (see Section \ref{sec:pricing_light_exotics}) and (2) qualitatively and quantitatively study the model spot-volatility dynamics, by analyzing the joint behavior of SPX and VIX. Note, moreover, that since the 4FPDV model is Markovian, simulating sample paths in this model is easier and faster than in non-Markovian models, so our Monte Carlo calibration is relatively fast (see Sections \ref{sec:calib_SPX_surface} and \ref{sec:joint}).
\end{remark}
}

\subsection{The learning procedure}

 The steps of the learning procedure are outlined as follows:
\begin{enumerate}
    \item For each of the 10 model parameters $\Theta_{i}$, sample $N>0$ independent configurations $\Theta_{i}(\omega_{j})$ accounting for the natural constraint $\Theta_{i}\in[a_{i},b_{i}]$:
    \begin{equation*}
        \forall\ i\in\{1,\dots,10\},\ \forall\ j\in\{1,\dots,N\}, \qquad \Theta_{i}(\omega_{j}) \sim \Ucal[a_{i},b_{i}],
    \end{equation*}
    where $\Ucal[a,b]$ denotes the uniform distribution over the interval $[a,b]$.
    \item For each parameter combination $j\in\{1,\dots,N\}$, simulate one random realization of the vector $R_{t}(\omega_{j}):=(R_{1,0,t}^{j},R_{1,1,t}^{j},R_{2,0,t}^{j},R_{2,1,t}^{j})$ in the 4FPDV model \eqref{eq:vol_dynamics} for a grid of future dates $t\in\{t_{1}(\omega_j),\dots,t_{M}(\omega_j)\}$ to be specified later.
    \item Also store the corresponding $\widehat{\textrm{VIX}}_{t}(\omega_{j})$ for $t\in\{t_{1}(\omega_j),\dots,t_{M}(\omega_j)\}$ and $j\in\{1,\dots, N\}$, where $\widehat{\VIX}$ denotes the square root of the random variable \eqref{initial_formula_vix_model_PDV} estimated using nested Monte Carlo.
    \item Randomly split the total number of parameter configurations into $N_{1}+N_{2}=N$. The first $N_{1}$ samples consist of the training set, while the next $N_{2}$ samples form the validation set.     
    \item  For a limited number of epochs, train a neural network $\Ncal\Ncal : (\Theta,R)\in \R^{14} \mapsto \textrm{VIX}\in \R^{+}$ on $N_{1}$ while tuning the architecture based on the validation set $N_{2}$ using a cross-validation approach, e.g., Random Search, Bayesian Optimization, or Hyperband; see \cite{LJDRT:17}. We use as loss function the root mean squared error (RMSE), i.e., the square root of
   { \begin{equation*}
        \frac{1}{NM}\sum_{j=1}^{N}\sum_{k=1}^{M}\left(\widehat{\textrm{VIX}}_{t_k(\omega_j)}(\omega_{j})-\mathcal{N}\mathcal{N}(\Theta(\omega_{j}),R_{t_k(\omega_j)}(\omega_{j}))\right)^2.
    \end{equation*}}
    The inputs of the network are standardized to help the training process.
\end{enumerate}
\begin{remark}
    We did not consider the loss function
{\begin{equation*}
    \frac{1}{NM}\sum_{j=1}^{N}\sum_{k=1}^{M}\left(\frac{1}{\Delta}\int_{t_k(\omega_j)}^{t_k(\omega_j)+\Delta}\sigma_{t}^2(\omega_{j})\,\d t-\mathcal{N}\mathcal{N}(\Theta(\omega_{j}),R_{t_k(\omega_j)}(\omega_{j}))\right)^2
\end{equation*}
(which does not require simulating nested paths) to learn the VIX squared} due to the very large conditional variance of the integrated instantaneous variance. Using averages $\widehat{\textrm{VIX}}_{t}(\omega_{j})$ improves the learning, as analyzed in \cite{ALL:23}.
\end{remark}

Note that this methodology naturally extends to any parametric Markovian volatility model. The above routine is first used to select the optimal architecture for the neural network. The training further continues to achieve higher accuracy. In the next sections, we report our numerical results.

\subsection{Training and testing of the network}\label{sec:training}
We consider $N_{1}=6.8\cdot 10^{5}$, $N_{2}=1.2\cdot10^{5}$ configurations of the model parameters to train and validate the network, respectively. Natural constraints yielding enough flexibility for the volatility smiles generated by the 4FPDV model \eqref{eq:vol_dynamics} are given by
\begin{align*}
 \beta_{0} & \in [0,0.2], \quad \beta_{1}\in [-0.25,0), \quad  \beta_{2}\in [0,1),\quad 
    \beta_{1,2}\in [0,0.3],\\
    \lambda_{n,p} & \in [1,100],  \quad
    \theta_{n}\in [0,1], \qquad n\in\{1,2\}, \quad p\in\{0,1\}.
\end{align*}
We require that $\lambda_{n,0}>\lambda_{n,1}$ so that $\lambda_{n,0}$ encodes the short memory and $\lambda_{n,1}$ the long memory. Notice that possible combinations of the above may give rise to extremely large values of the instantaneous volatility of the instantaneous volatility, hence in the simulation of the training set we remove the combinations for which $ |\beta_{1}|\bar{\lambda}_{1} > 10$. We observe the four factors $R_t(\omega_j)$ and the VIX, namely $\widehat{\textrm{VIX}}_{t}(\omega_j)$, at $M=200$ evenly spaced time points $t_1(\omega_j)<\cdots<t_M(\omega_j)$, with $t_{1}(\omega_j)= \max\{1/\lambda_{1,0}(\omega_j),1/\lambda_{2,0}(\omega_j)\}$ and $t_{M}(\omega_j)=1$. {\cb With this choice of $t_1(\omega_j)$, the four factors are often observed close to or in their stationary states, as we numerically checked.\footnote{{\cb For a stronger guarantee, one could choose $t_{1}(\omega_j)= \max\{1/\lambda_{1,1}(\omega_j),1/\lambda_{2,1}(\omega_j)\}$ and use a larger value of $t_M(\omega_j)$. The results of our learning procedure show that our choice of $t_1(\omega_j)$ and $t_M(\omega_j)$ is good enough to sample a large variety of values of $R$.}}} We compute $\widehat{\VIX}_{t}$ by nested Monte Carlo: $10^{4}$ nested paths are used to simulate the conditional expectation in \eqref{initial_formula_vix_model_PDV}, where the integral is approximated with Riemann sums. Both processes $\sigma$ and $R$ are sampled with discretization step $\Delta t=\frac{1}{2520}$.

To assess the goodness of the training, we consider two tests. First, we build the histogram of the average absolute error over 10,000 different Monte Carlo samples between the optimal neural network approximation $\Ncal\Ncal^{\star}$ and $\widehat{\VIX}$ for 1,000 random configurations of the model parameters. 
In Figure \ref{fig:Hist_0}, we observe that for 99\% of the tested configurations, the mean absolute error is smaller than 0.55 (around half a volatility point).
    \begin{figure}[H]
         \centering
         \includegraphics[scale=0.38]{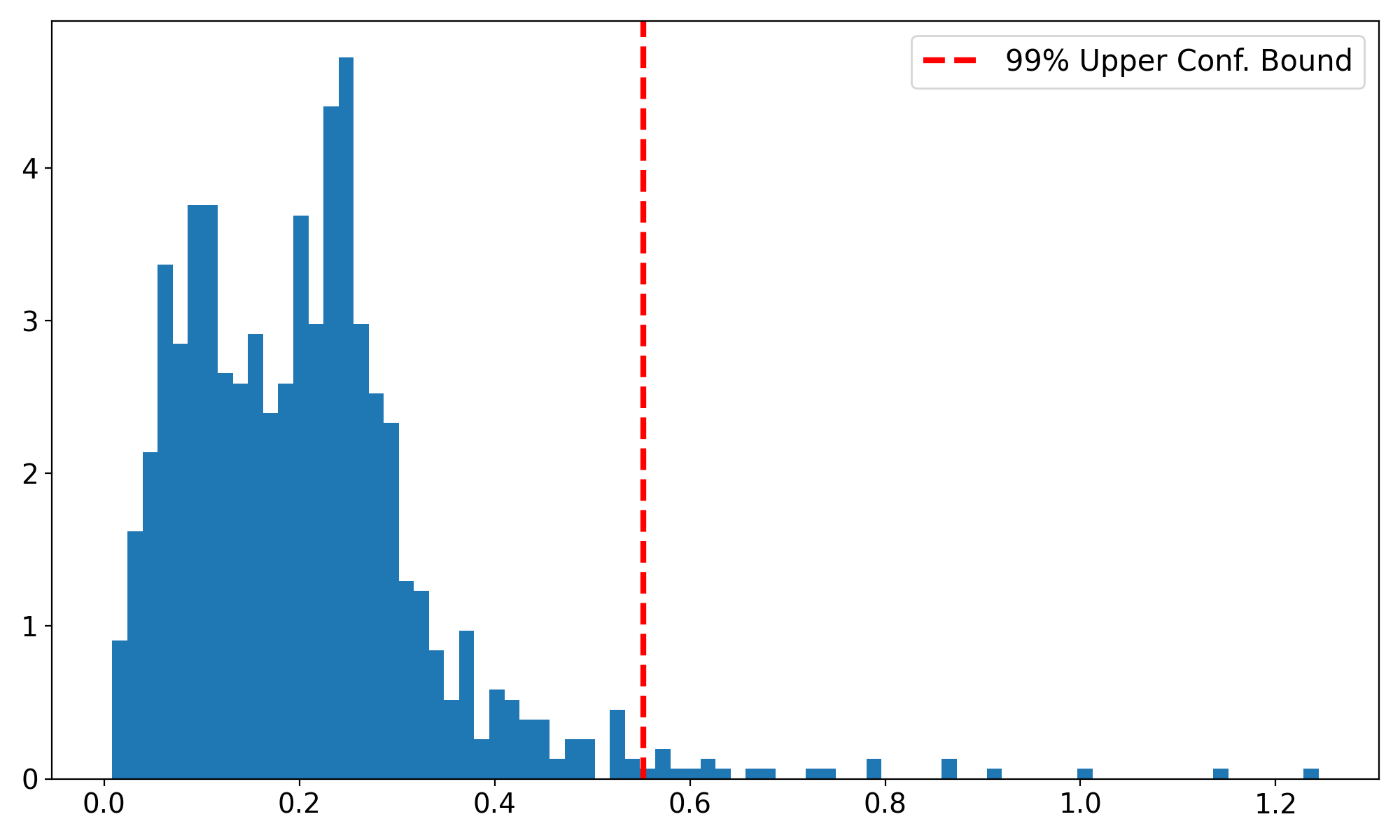}
         \caption{Histogram of the average absolute error between $\Ncal\Ncal^{\star}$ and $\widehat{\VIX}$ for $10^{3}$ configurations of parameters of the model. The average is taken over $10^{4}$ Monte Carlo samples.}
         \label{fig:Hist_0}
    \end{figure}
    Some of the 1,000 configurations of model parameters tested in Figure \ref{fig:Hist_0} generate unrealistic VIX smiles and may produce the largest average absolute errors. 
Therefore, in our second test, we compare the VIX implied volatility smile computed using the neural approximation $\Ncal\Ncal^{\star}$ of the VIX with the one computed using $\widehat{\VIX}$ for a few sets of market-calibrated parameters. Figure \ref{fig:IV1} shows an excellent agreement between the two VIX smiles and the two VIX futures for one set of realistic parameters.

\begin{figure}
         \centering
         \includegraphics[trim={0.7cm 0.0cm 0.7cm 0.7cm},clip,scale=0.33]{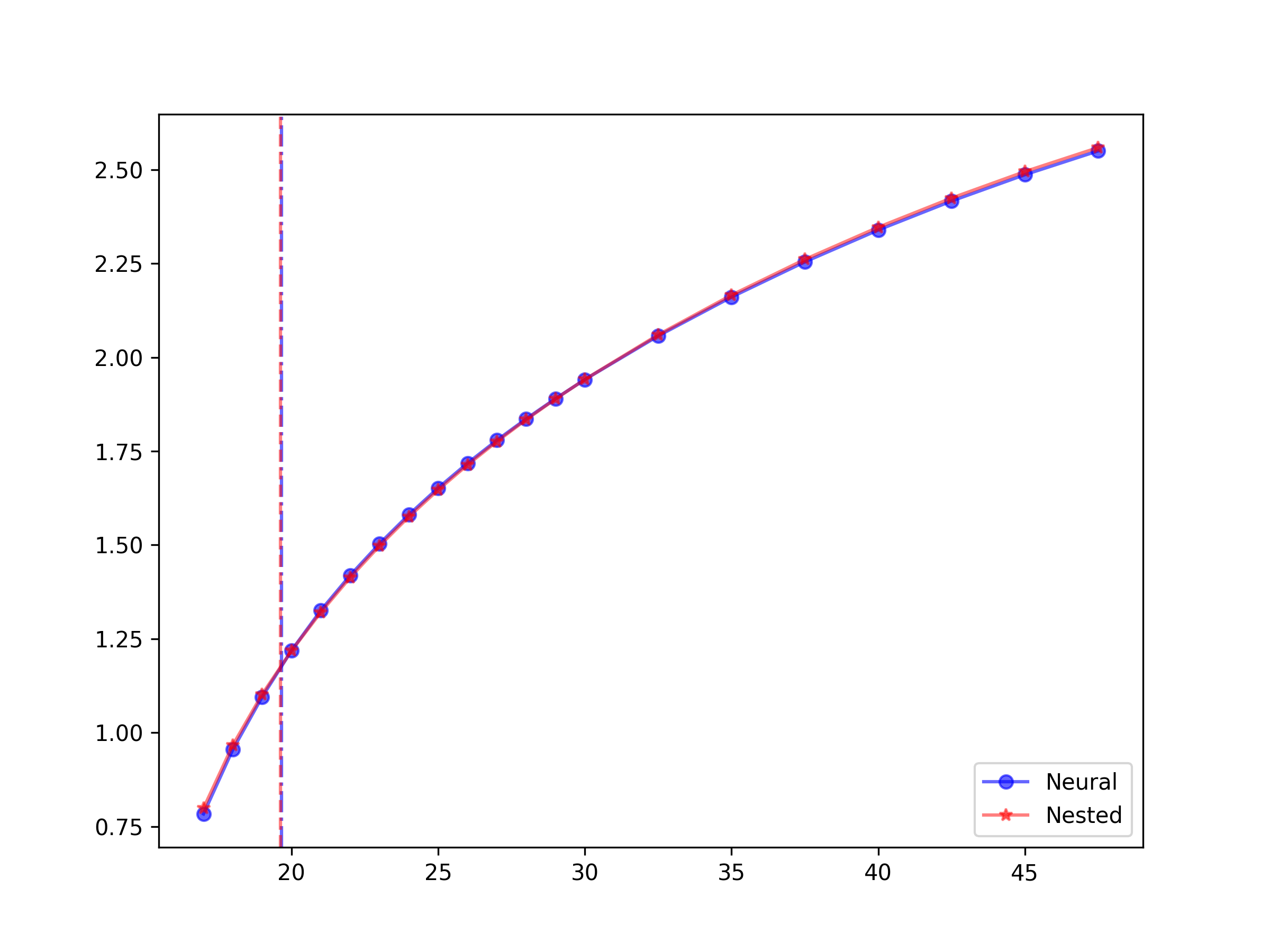}
         \includegraphics[trim={0.7cm 0.0cm 0.7cm 0.7cm},clip,scale=0.25]{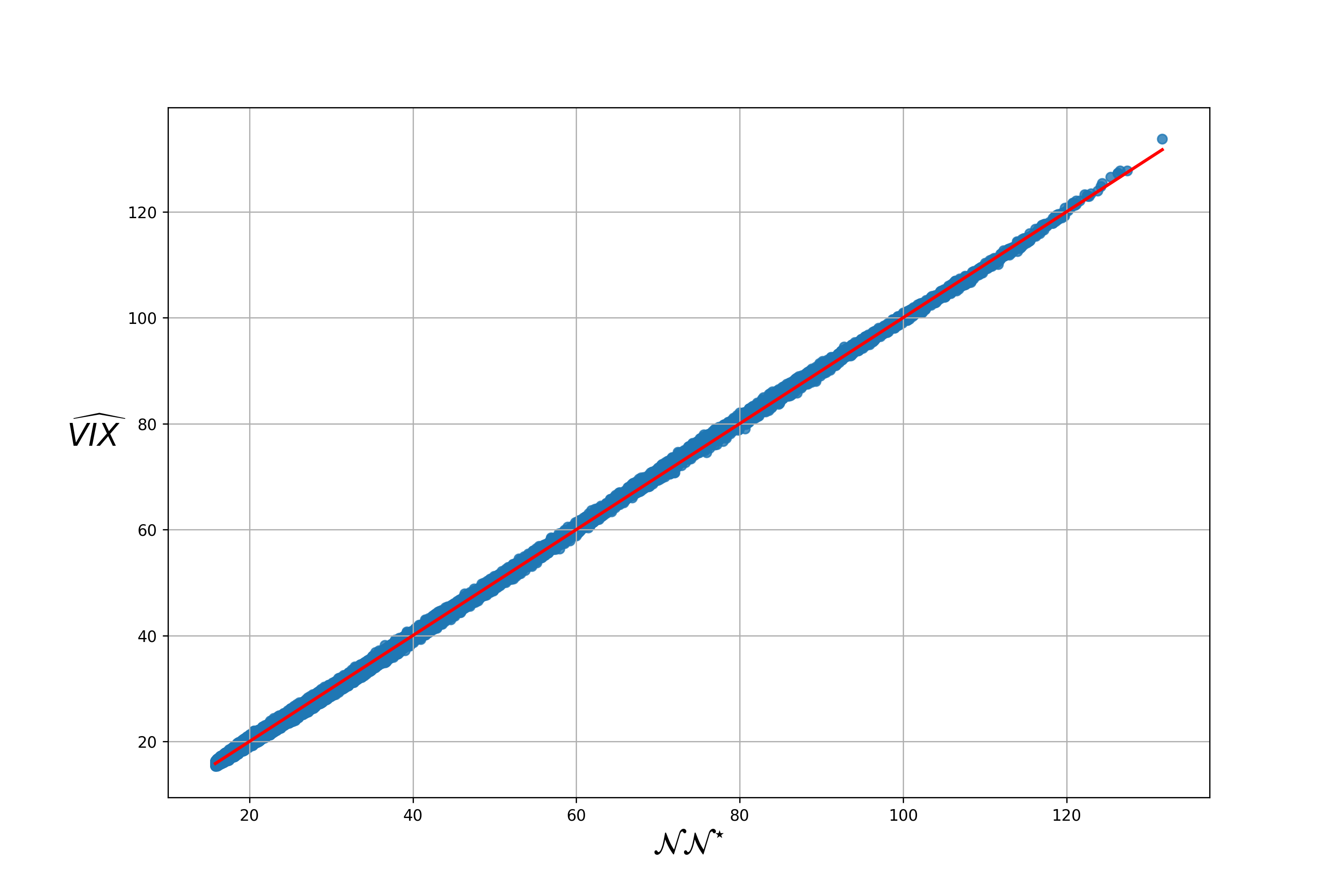}\\
         \includegraphics[trim={0.3cm 0.4cm 0.5cm 0.0cm},clip,scale=0.34]{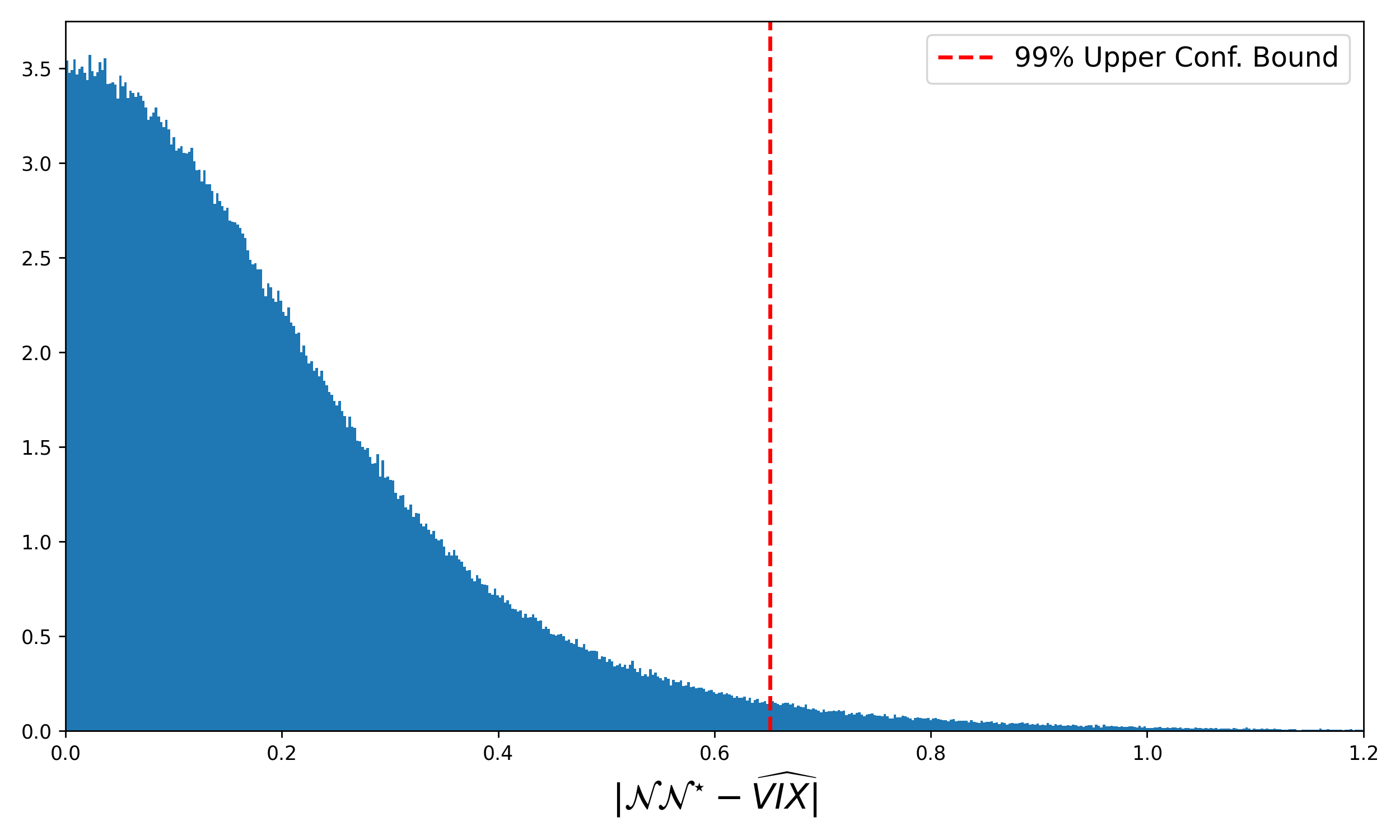}
         \caption{Model parameters: $\lambda_{1,0}=29$, $\lambda_{1,1}=20$, $\theta_1=0.69$, $\lambda_{2,0}=81 $, $\lambda_{2,1}=66$, $\theta_2=0.25$, $\beta_{0}=0.11$, $\beta_{1}=-0.057$, $\beta_{2}=0.1$, $\beta_{1,2}=0.256$. Top left: VIX smiles and VIX future (dashed) with maturity 16 days. In blue: the VIX is computed via the pre-trained neural network $\Ncal\Ncal^{\star}$.
         In red: the VIX is computed via nested Monte Carlo, $\widehat{\VIX}$. Top right: A total of $10^{6}$ predictions of $\Ncal\Ncal^{\star}$ vs. $\widehat{\VIX}$. Bottom: histogram of the absolute error between $\Ncal\Ncal^{\star}$ and $\widehat{\VIX}$.}
         \label{fig:IV1}
\end{figure}

In Figure \ref{fig:IV1} we additionally plot the predictions of the neural network against the nested Monte Carlo benchmark as well as the histogram of the absolute error between the two. For this particular set of realistic parameters, the average absolute error is 0.2, which is in the bulk of the distribution of average absolute errors plotted in Figure \ref{fig:Hist_0}. For 99\% of the $10^6$ Monte Carlo samples, the (pathwise) absolute error between $\Ncal\Ncal^{\star}$ and $\widehat{\VIX}$ is smaller than 0.65.

In Appendix \ref{appendix:nn}, we report similar results for parameters jointly calibrated to the SPX and VIX smiles on June 2, 2021, and June 3, 2021. This informs us that for market-calibrated parameters the optimal neural network accurately learns the VIX computed with nested Monte Carlo and can be used to price  VIX futures and VIX options. In any case, when we calibrate model parameters to VIX market data using the $\Ncal\Ncal^{\star}$ approximation of the VIX, we systematically check that the VIX smile computed using $\Ncal\Ncal^{\star}$ and the one computed using the nested estimator $\widehat{\VIX}$ match each other closely (see Figures \ref{fig:joint2} and \ref{fig:joint3}).

\subsection{The architecture of the neural network}\label{sec:architecture}

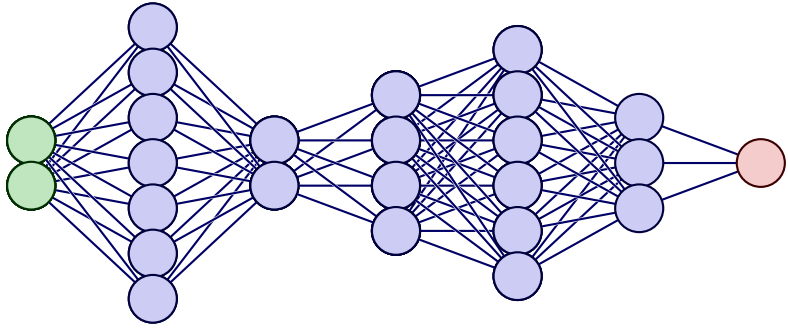
\begin{figure}
  \centering
    \begin{tikzpicture}[x=1.6cm,y=0.6cm]
  \message{^^JNeural network large}
  \readlist\Nnod{2,7,2,4,6,3,1} 
  
  \message{^^J  Layer}
  \foreachitem \N \in \Nnod{ 
    \def\lay{\Ncnt} 
    \pgfmathsetmacro\prev{int(\Ncnt-1)} 
    \message{\lay,}
    \foreach \i [evaluate={\y=\N/2-\i; \x=\lay; \n=\nstyle;
                           \nprev=int(\prev<\Nnodlen?min(2,\prev):3);}] in {1,...,\N}{ 
      
      \coordinate (N\lay-\i) at (\x,\y) {};
      
      \ifnum\lay>1 
        \foreach \j in {1,...,\Nnod[\prev]}{ 
          \draw[connect,white,line width=1.2] (N\prev-\j) -- (N\lay-\i);
          \draw[connect] (N\prev-\j) -- (N\lay-\i);
          \node[node \nprev,minimum size=18] at (N\prev-\j) {}; 
        }
        \ifnum \lay=\Nnodlen 
          \node[node \n,minimum size=18] at (N\lay-\i) {};
        \fi
      \fi 
      
    }

  }
\end{tikzpicture}
  \caption{Example of architecture for a feed-forward neural network.} 
\end{figure}

Here, we report the architecture of the network obtained after the hyperparameter tuning via Bayesian optimization. We asked the KerasTuner to select from 1 to 5 dense hidden layers for a feed-forward neural network, with a number of nodes ranging from 64 to 512 with 32 as step size, and two possible activation functions, namely 
\begin{equation*}
    \text{ReLU}(x):=\max\{x,0\}, \quad \text{tanh}(x):=\frac{e^x-e^{-x}}{e^x+e^{-x}}, \qquad\forall x\in\R.
\end{equation*}
 We report the optimal values in Table 2. Note that the optimal architecture uses the full 5 hidden layers, with varying numbers of nodes per layer, and that it mixes both activation functions. Moreover, we tuned the learning rate of the Adam optimizer between $10^{-5}$ and $10^{-2}$ with log-sampling. The optimal learning rate is $4.2 \cdot 10^{-5}$.
\begin{longtable}{|c|c|c|}
\caption{\label{long} Architecture {of our neural network}}\\
\hline
Layer type & Activation & No.\ Nodes\\
\hline
\endfirsthead
\hline
Layer type & Activation & Nbr. Nodes\\
\hline
\endhead
\hline
\endfoot
\hline

\endlastfoot
Input & ReLU & 14 \\
Hidden 1 & tanh & 448\\
Hidden 2 & tanh & 64\\
Hidden 3 & ReLU & 224\\
Hidden 4 & tanh & 416\\
Hidden 5 & ReLU & 128\\
Output  & linear & 1\\
\end{longtable}

{\cb \subsection{Joint SPX/VIX sample paths}\label{sec:sample_paths}
Since we learn the VIX in the 4FPDV model \emph{pathwise}, we can very quickly simulate VIX paths in the model along with SPX paths using $\mathcal{N}\mathcal{N}^\star$. This is very useful to investigate the dynamical pathwise properties of the model.
For instance, in Figure \ref{fig:sample_paths}, we consider parameters used in \cite{GLJ:22} to produce realistic paths. We check that VIX spikes correspond to recent drops in the SPX, as observed in the market. Thanks to the $R_2$ factors (the memory of volatility), the model also reproduces the time-asymmetry of large VIX spikes observed in the market, i.e., the fact that for large VIX spikes, the VIX tend to spike very quickly and decrease more slowly when the market recovers from a large drop. This is due to the fact that both the $R_1$ and $R_2$ factors push the volatility up when the SPX is down (a series of negative Brownian increments $dW_t$ is drawn), while they have opposite effects on volatility when the market recovers. In particular, the long memory of volatility factor $R_{2,1}$ tends to keep volatility large for a while, as it remembers large square returns that have happened in the past weeks or months. In Figure \ref{fig:sample_paths}, we also display a realization of the path of the VIX computed via the nested Monte Carlo, i.e., $\widehat{\VIX}$. As we already knew from the tests reported in Section \ref{sec:training}, sampling via $\mathcal{N}\mathcal{N}^\star$ is very accurate. Moreover, sampling via $\mathcal{N}\mathcal{N}^\star$ is of course much faster: it is more than 2,400 times faster than via the nested Monte Carlo with $10^4$ inner simulations: approximately 0.09 seconds vs. 4.3 minutes (on a local computer with Core i7).}

\begin{figure}
    \centering
    \includegraphics[scale=0.42]{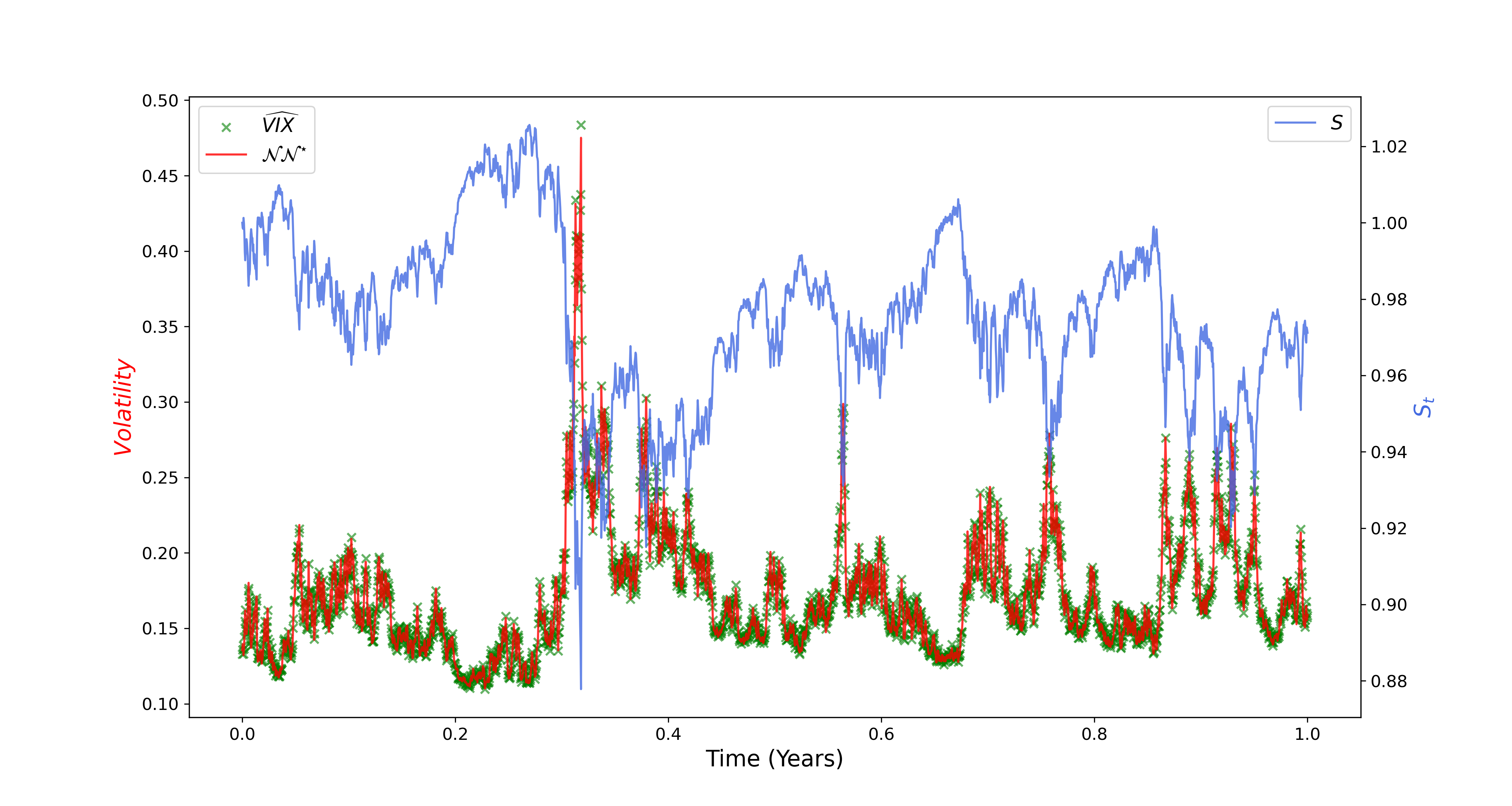}
    \caption{\cb Sample paths of SPX and VIX in the 4FPDV model. In red: approximation of the pre-trained neural network $\mathcal{N}\mathcal{N}^\star$; green markers: nested Monte Carlo estimator $\widehat{\VIX}$; in blue: SPX. Model parameters: $\lambda_{1}=(55,10)^\top$, $\lambda_{2}=(20,3)^\top$, $\beta=(0.04,-0.13,0.65)^\top$, $\beta_{1,2}=0$, $\theta_{1}=0.25$, $\theta_{2}=0.5$, $S_{0}=1$. Initial values of the Markovian factors: $R_{1,0,0}= 0.078$; $R_{1,1,0} = 0.16$ $R_{2,0,0}=0.074$, $R_{2,1,0} =0.016$. Discretization step: $\Delta t=\frac{1}{2520}$.}
    \label{fig:sample_paths}
\end{figure}

\section{Calibration to SPX options}\label{sec:calib_SPX}

{\cb We now turn to the issue of calibrating the model to option prices. We first consider SPX options only; in Section \ref{sec:joint} we will add VIX options, using the neural network approximation built in Section \ref{sec:NN}.}

\subsection{Pricing and calibration of SPX options}

Suppose we are given a set $\Tcal^{\SPX}$ of SPX option maturities. Then, for each $T\in \Tcal^{\SPX}$, we consider a collection $\Kcal_T^{\SPX}$ of strike prices. Under model \eqref{model} with parameters $\Theta$, the price of an SPX call or put option is approximated by
\begin{align}
    C_{\Theta}^{\SPX}(T,K)&=e^{-\int_{0}^{T}r_s\,\d s}\frac{1}{N_{MC}}\sum_{j=1}^{N_{MC}}(S_{T}(\omega_{j})-K)^{+}, \label{eqz:mc_model_price_call} \\
    P_{\Theta}^{\SPX}(T,K)&=e^{-\int_{0}^{T}r_s\,\d s}\frac{1}{N_{MC}}\sum_{j=1}^{N_{MC}}(K-S_{T}(\omega_{j}))^{+}, \label{eqz:mc_model_price_put}
\end{align}
for all $ T\in\Tcal^{\SPX}$ and for all 
 $K\in\Kcal_T^{\SPX}$ respectively. Denote by $\sigma_{\text{IV},\Theta}^{\SPX}(T,K)$ the model implied volatility under the parameters $\Theta$ and for a fixed maturity and strike price, $T,K>0$ computed for OTM options. Since by construction the model and market SPX futures curve agree we can consider a loss function which aims to measure the discrepancy between either option prices (see, e.g., \cite{CBH:04,PS:14,CGMS:23}) or implied volatilities (as, e.g., in \cite{PPR:18,AJIL:22,BPS:22,GM:22}). In the following we consider the latter approach and introduce the loss function 
  \begin{equation}\label{eq:Loss_SPX}
     L_{\SPX}(\Theta)=\omega_{\SPX}\frac{1}{\# \Tcal^{\SPX}}\sum_{T\in\Tcal^{\SPX}}\frac{1}{\# \Kcal_T^{\SPX}}\sum_{K\in \Kcal_T^{\SPX}}\ell\left(\sigma_{\text{IV},\Theta}^{\SPX}(T,K),\sigma_{\text{IV}}^{\SPX}(T,K)\right),
\end{equation}
 where $\omega_{\SPX}>0$ is a hyperparameter weight {\cb that will be useful when we address the joint SPX/VIX calibration problem} and {\cb $\ell:\mathbb{R}^{+}\times\mathbb{R}^{+}\to \mathbb{R}^{+}$ is the score function 
 \begin{equation}\label{eqz:score}
        \ell(x,y)=\left(\frac{x}{y}-1\right)^2.
    \end{equation}}

\subsection{Calibration to the SPX surface}\label{sec:calib_SPX_surface}
In Figures \ref{fig:spx_surface1} and \ref{fig:spx_surface2} we report the results of two calibration exercises: as of June 3, 2021, we calibrate to monthly SPX options with maturities between 15 and 351 days; and as of October 25, 2023, to monthly SPX options with maturities from 23 days to 296 days. At each iteration of the optimizer we simulate $N_{MC}= 2\cdot 10^{5}$ trajectories with discretization time step $\Delta t=\frac{1}{504}$. {\cb We have carefully checked that this time step is small enough to produce accurate prices and implied volatilities, even in the case of very large mean reversion ($\lambda_{1,0}=\lambda_{1,1}=100$) and vol-of-vol ($|\bar\lambda_1 \beta_1|=10$). Note that a mean reversion speed of 100 corresponds to a characteristic time of mean reversion of $\frac{1}{100}$, which is five times larger than our time step.} The calibrated parameters of the path-dependent volatility model are reported in Tables \ref{tab:sample2} and \ref{tab:sample1}, respectively.
{\cb It is remarkable that the model, which is low-parametric (10 parameters) and time-homogeneous (no deterministic time-dependent input curve is used to calibrate to a term-structure of volatilities) is able to fit SPX smiles for a wide range of maturities so accurately}. In particular, the two-exponential kernels allow us to accurately capture the market term-structure of the ATM skew (see Figure \ref{fig:ATM_Skew}). We observe that two different mean reversion speeds need to be coupled, see the values of the $\lambda$s and $\theta$s in Table \ref{tab:sample2}.

\begin{figure}
    \centering
    \includegraphics[width=\textwidth]{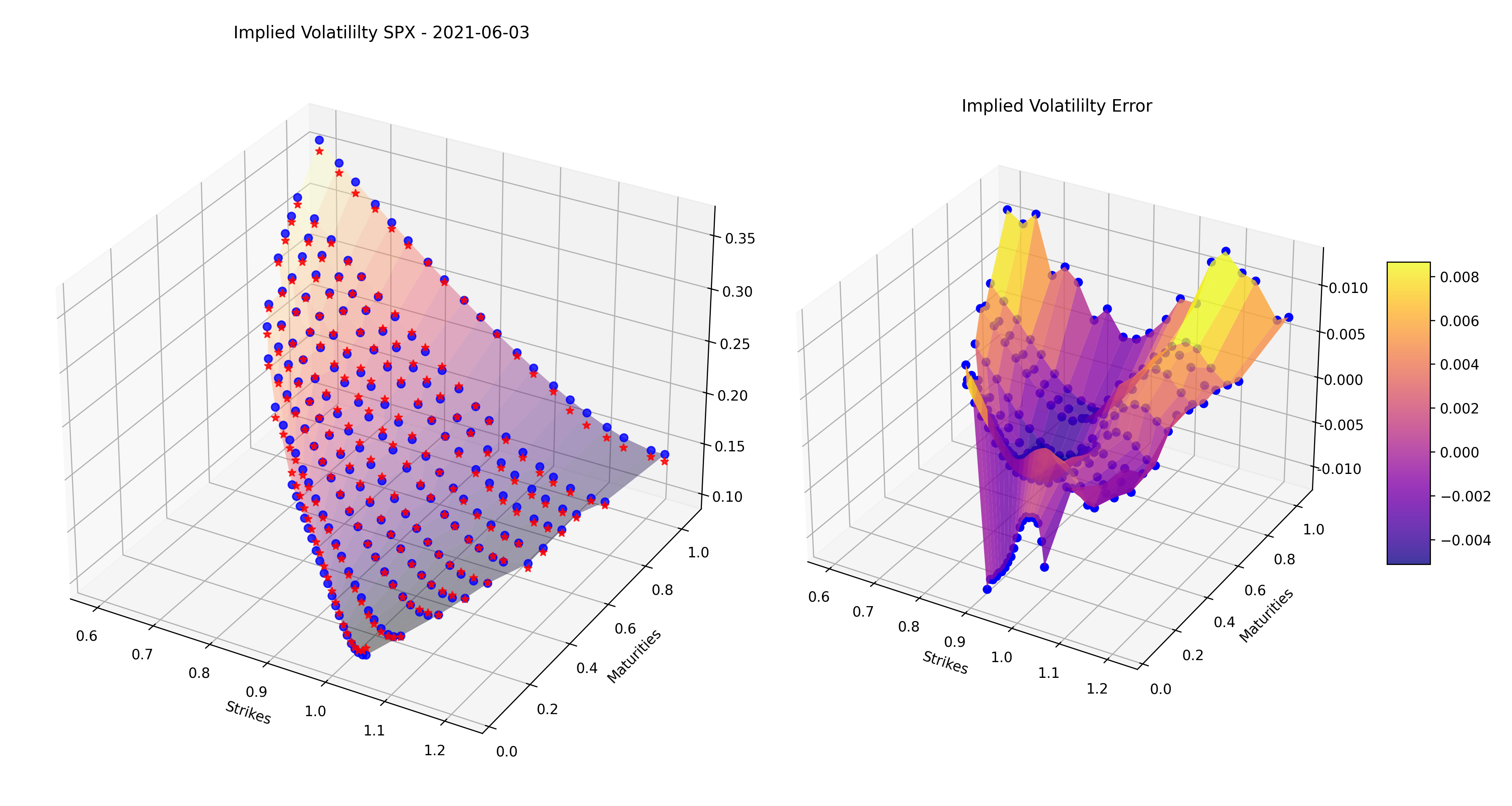}
    \caption{ Calibration as of June 3, 2021. Left: in blue dots, the calibrated implied volatility smiles under the PDV model; in red stars, the market mid implied volatility smiles.
    Right: the interpolated relative error between the calibrated implied volatility smiles and the market mid ones displayed on the left. We report the initial values of the factors given the calibrated parameters in the Table \ref{tab:sample2}. Recall that these are uniquely determined by the calibrated $\lambda$s and the observed past daily returns for June 3, 2021 with a cut-off at 1000 days: $R_{1,0}=0.0894$, $R_{1,1}=-0.1602$, $R_{2,0}=0.0031$, $R_{2,1}=0.0476$. 
  }
    \label{fig:spx_surface1}
\end{figure}

\begin{table}
    \begin{center} 
    \caption{Calibrated parameters of the 4FPDV model. Calibration to the SPX surface as of June 3, 2021.}
\label{tab:sample2}
\begin{tabular}{||c |c|c| c|c|c||}
 \hline
 $\lambda_{1,0}=34.39$ & $\lambda_{1,1}=13.26$ & $\theta_1=0.501$ & $\lambda_{2,0}=95.63$ & $\lambda_{2,1}=1.428$   & $\theta_2=0.448$\\ 
 \hline
\end{tabular}

\begin{tabular}{|| c |c|c| c||} 
 \hline
  $\beta_{0}=0.0493$ & $\beta_{1}=-0.1999$ & $\beta_{2}=0.5479$ & $\beta_{1,2}=0.2285$\\ 
 \hline
\end{tabular}
\end{center}

\end{table}

\begin{figure}
    \centering
    \includegraphics[width=\textwidth]{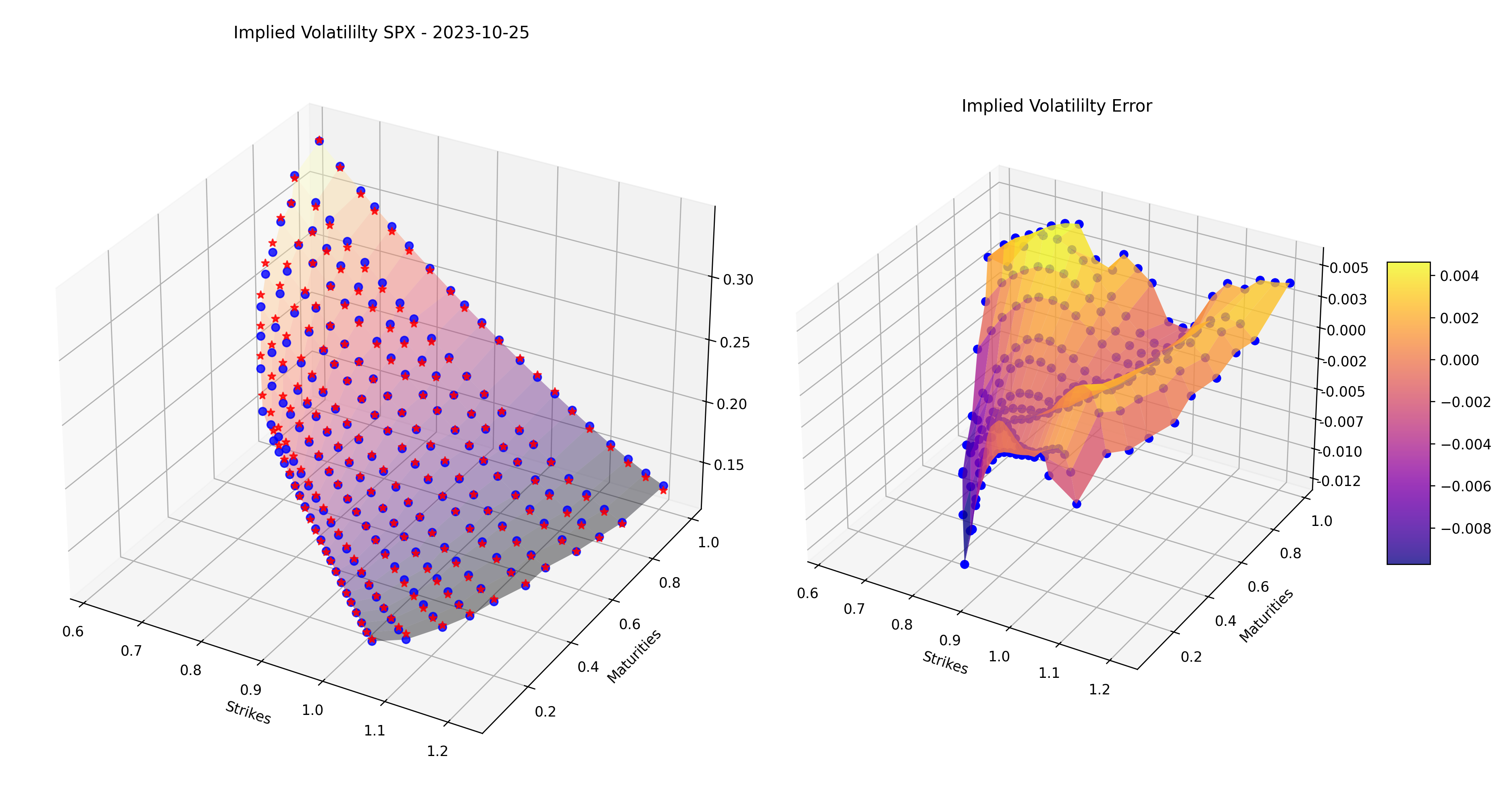}
    \caption{{ Calibration as of October 25, 2023. See caption of Figure \ref{fig:spx_surface1} for color reference.} The initial values of the factors  with a cut-off at 1000 days are given by: $R_{1,0}=-1.2054$, $R_{1,1}=-0.2428$, $R_{2,0}=0.0178$, $R_{2,1}=0.0166$.}
    \label{fig:spx_surface2}
\end{figure}

\begin{table}
    \begin{center} 
    \caption{{ Calibrated parameters of the 4FPDV model. Calibration to the SPX surface as of October 25, 2023.}}
\label{tab:sample1}
\begin{tabular}{||c |c|c| c|c|c||}
 \hline
 $\lambda_{1,0}=53.03$ & $\lambda_{1,1}=6.031$ &$\theta_1=0.685$  & $\lambda_{2,0}=12.03$ & $\lambda_{2,1}=8.325$  & $\theta_2=0.2876$\\ 
 \hline
\end{tabular}

\begin{tabular}{|| c |c|c| c||} 
 \hline
  $\beta_{0}=0.0381$ & $\beta_{1}=-0.1483$ & $\beta_{2}=0.7097$ & $\beta_{1,2}=0.1671$\\ 
 \hline
\end{tabular}
\end{center}

\end{table}

\begin{figure}
  \centering

  \begin{subfigure}{0.49\textwidth}
    \includegraphics[width=\linewidth]{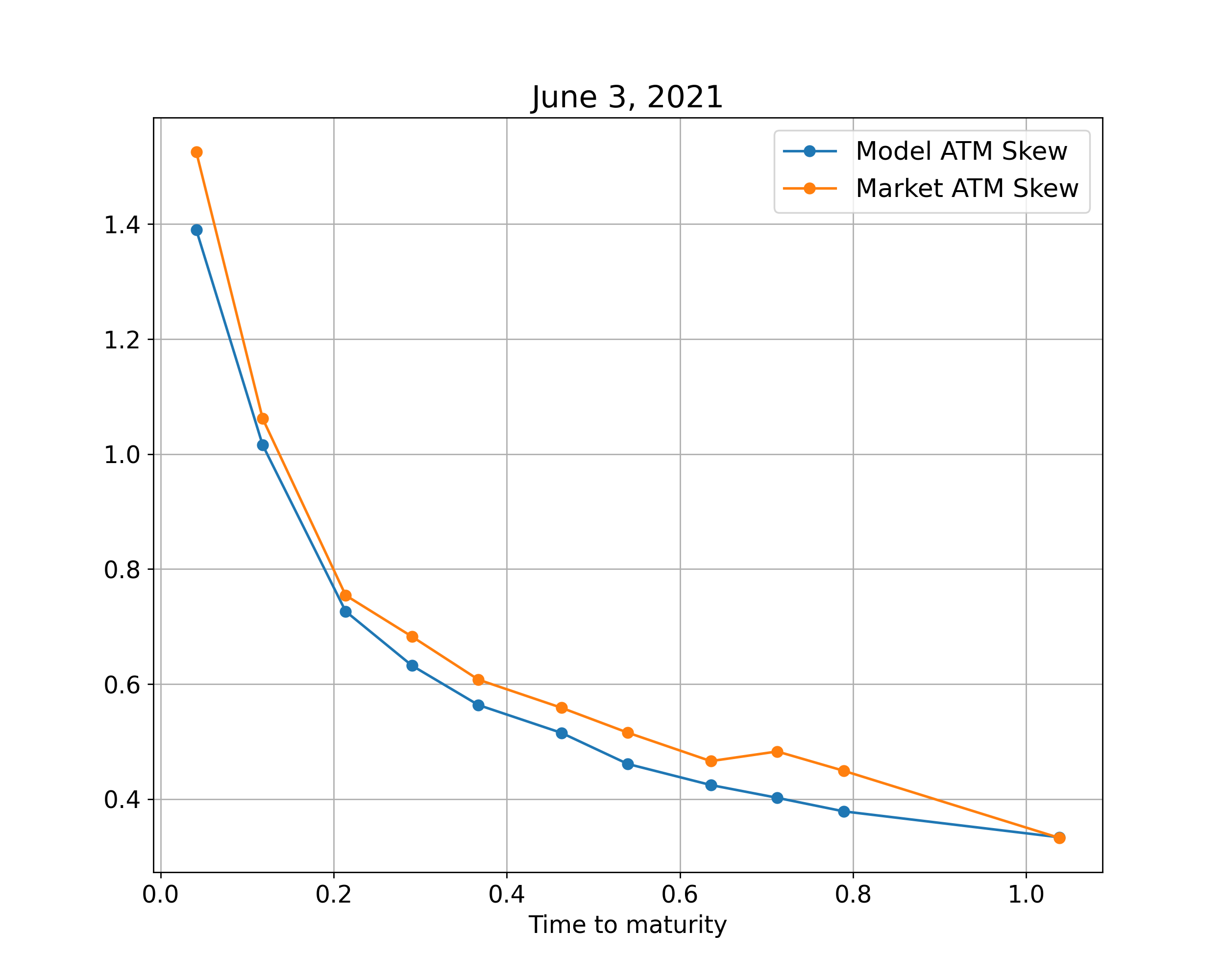}
    \caption{}
  \end{subfigure}
  \hfill
  \begin{subfigure}{0.49\textwidth}
    \includegraphics[width=\linewidth]{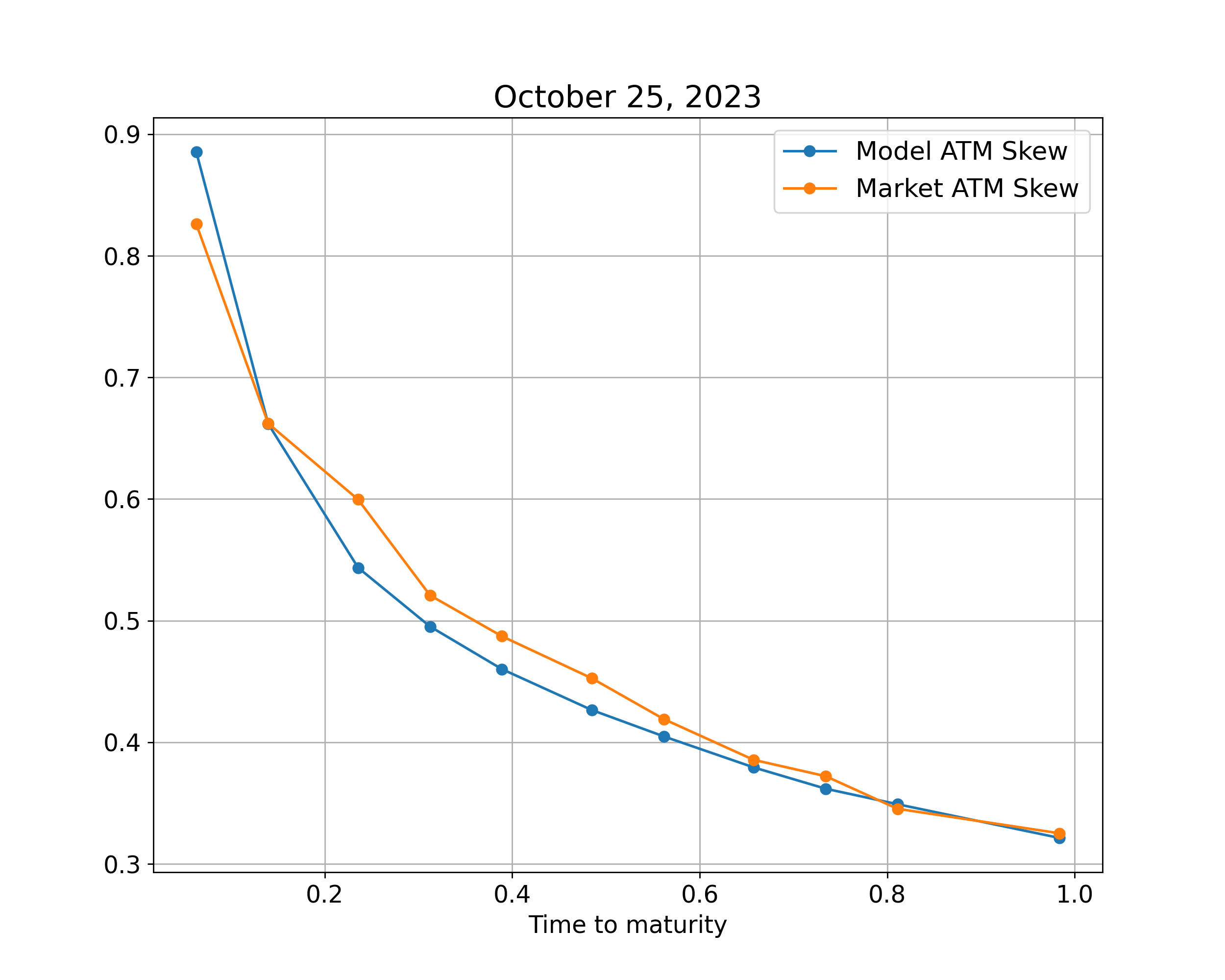}
    \caption{}
  \end{subfigure}
  
 \caption{Comparison of the market and model term-structures of ATM skew.}
 \label{fig:ATM_Skew}
\end{figure}

The calibration to the SPX surface {up to the one-year maturity} takes approximately 30 minutes using Cuda 11.4 with PyTorch on a GPU-NVIDIA-SMI. The optimization was carried out, starting from a randomized offline initial configuration, using the derivative-free optimizer Py-BOBYQA (see \cite{CFMR:19,CRS:22}) relying on a trust-region method which proved particularly efficient with respect to other optimizers. 
{ We enforced natural bounds on the parameters during optimization and then scaled model parameters to help the learning.}

\section{Joint calibration of SPX and VIX options}\label{sec:joint}

\subsection{Pricing and calibration of VIX futures and VIX options}

\noindent{\bf Pricing.} Suppose we are given a set $\Tcal^{\VIX}$ of VIX option maturities. For each $T\in \Tcal^{\VIX}$, we consider a collection $\Kcal_T^{\VIX}$ of strike prices. {\cb We use the classical Monte Carlo method together with the neural approximation $\Ncal\Ncal^{\star}$ of the VIX under model \eqref{eq:vol_dynamics} with parameters $\Theta$ built in Section \ref{sec:NN}.} The price of the VIX future and the price of a VIX call are respectively approximated by
\begin{align*}
F_{\Theta}^{\VIX}(T)&=\frac{1}{N_{MC}}\sum_{j=1}^{N_{MC}}\Ncal\Ncal^{\star}(\Theta,R_{T}(\omega_{j})),\\
    C_{\Theta}^{\VIX}(T,K)&=e^{-\int_{0}^{T}r_s\,\d s}\frac{1}{N_{MC}}\sum_{j=1}^{N_{MC}}(\Ncal\Ncal^{\star}(\Theta,R_{T}(\omega_{j}))-K)^{+},
\end{align*}
for all $ T\in\Tcal^{\VIX}$ and for all 
 $K\in\Kcal_T^{\VIX}$.

 \medskip

 \noindent{\bf Calibration.} We now turn to the problem of calibrating the model parameters to VIX futures and VIX options. Notice that since the actual tradable underlying of VIX options are the VIX futures with the same expiry as the option, it is crucial that the VIX futures be very well calibrated, as already argued in, e.g., \cite{PPR:18,G:20,CGMS:23}. To be consistent, the model VIX implied volatilities are then of course computed by inverting the Black formula using the corresponding model VIX future.
 
 Since we have learned the VIX not only as a function of the four Markovian factors but also as a function of the 10 model parameters, $\Theta$, we can use the loss function 
 \begin{align}\label{eq:Loss_VIX}
     L_{\VIX}(\Theta) &=\omega_{F}\frac{1}{\# \Tcal^{\VIX}}\sum_{T\in\Tcal^{\VIX}}\ell\left(F_{\Theta}^{\VIX}(T),F^{\VIX}(T)\right)\\ \nonumber
     & \quad +\omega_{\VIX}\frac{1}{\# \Tcal^{\VIX}}\sum_{T\in\Tcal^{\VIX}}\frac{1}{\# \Kcal_T^{\VIX}}\sum_{K\in \Kcal_T^{\VIX}}\gamma_{T,K}^{\VIX}\ell\left(C_{\Theta}^{\VIX}(T,K),C^{\VIX}(T,K)\right)
 \end{align}
 to calibrate the model to VIX futures and VIX options, where 
 \begin{itemize}
     \item $\omega_{F},\ \omega_{\VIX}>0$ are hyperparameters weights;
     \item  $\#$ denotes the cardinality of a set;
     \item {$(\gamma_{T,K}^{\VIX})_{T,K}$ are normalized vega weights; given $\Kcal(T)$ the strikes available for maturity $T>0$,
     \begin{equation*}
         \gamma_{T,K}:=\frac{\Vcal_{T,K}}{\sum_{K'\in\Kcal(T)}\Vcal_{T,K'}},
     \end{equation*}
     where $\Vcal_{T,K}$ denotes the Black vega of the call option with maturity $T$ and strike $K$, computed using the corresponding market mid implied volatility.}
     \item  $F^{\VIX}, C^{\VIX}$ denote the market futures and the market call option prices, respectively;
     \item  $\ell$ is the score function \eqref{eqz:score}.
 \end{itemize}
Note that we calibrate the model to VIX option prices, not VIX implied volatilities, as model VIX implied volatilities depend on model VIX futures; this prevents a mismatch between model and market VIX futures to cascade into a mismatch between model and market VIX implied volatilities. The weights $(\gamma_{T,K}^{\VIX})_{T,K}$ are often taken to be the inverse of the market bid-ask spreads, which gives more weight to those options where the bid-ask spread is tighter (see, e.g., \cite{CBH:04} for a discussion on this topic). In Figure \ref{fig:vega} we compare the normalized vega weights with the inverse of the market bid-ask spreads as of June 2, 2021, for VIX options with 33 days to expiration. The vega weights are a smoothed version of the inverse of the bid-ask spread.

\begin{figure}
    \centering
\includegraphics[trim = 0 0.7cm 0 1cm, clip, scale=0.41]{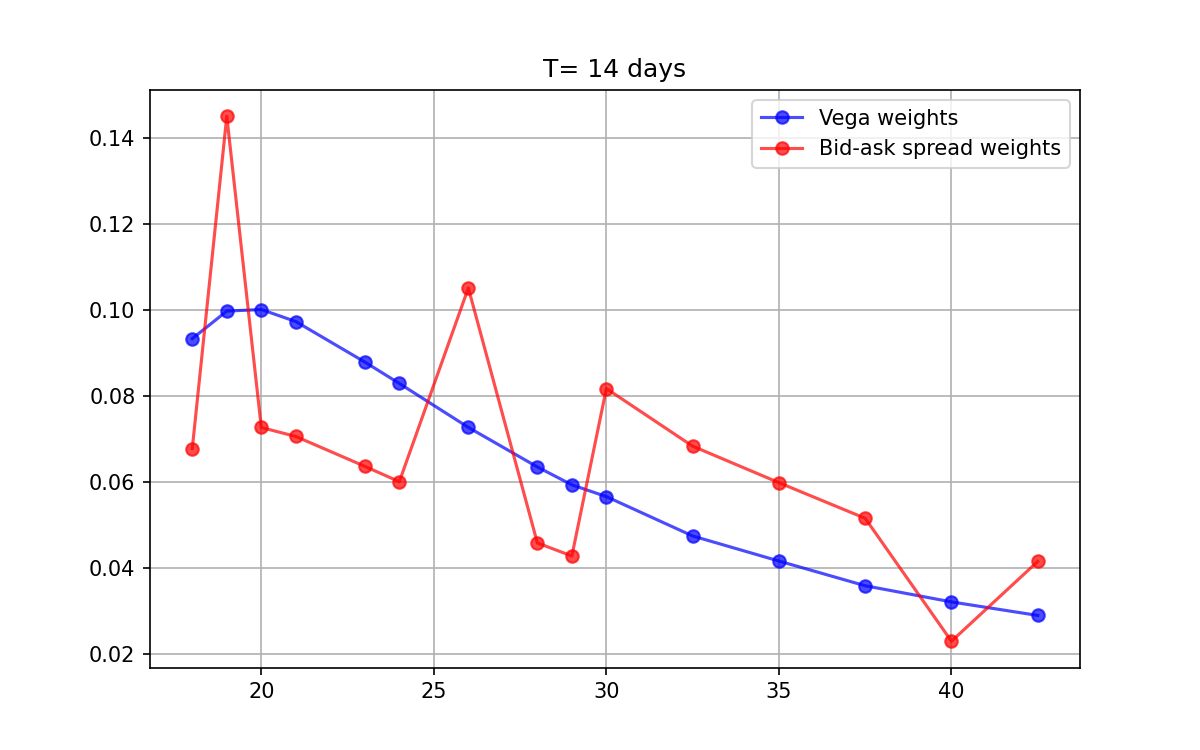}
    \caption{In blue: normalized vega weights; in red: normalized bid-ask spread weights. Data as of June 2, 2021.}
    \label{fig:vega}
\end{figure}

\subsection{Joint calibration results}

We now report the results of the joint calibration of SPX and VIX monthly options as of June 2, 2021 (Figure \ref{fig:joint2} and Table \ref{tab:joint11}) and June 3, 2021 (Figure \ref{fig:joint3} and Table \ref{tab:joint1}). As loss function we take the sum of the two loss functions \eqref{eq:Loss_SPX} and \eqref{eq:Loss_VIX}. Recall that SPX monthly options expire the third Friday of the expiration month, while the monthly VIX futures and VIX options mature 30 days before the third Friday of the following month.

\begin{figure}[H]
\centering
         \includegraphics[width=\textwidth]{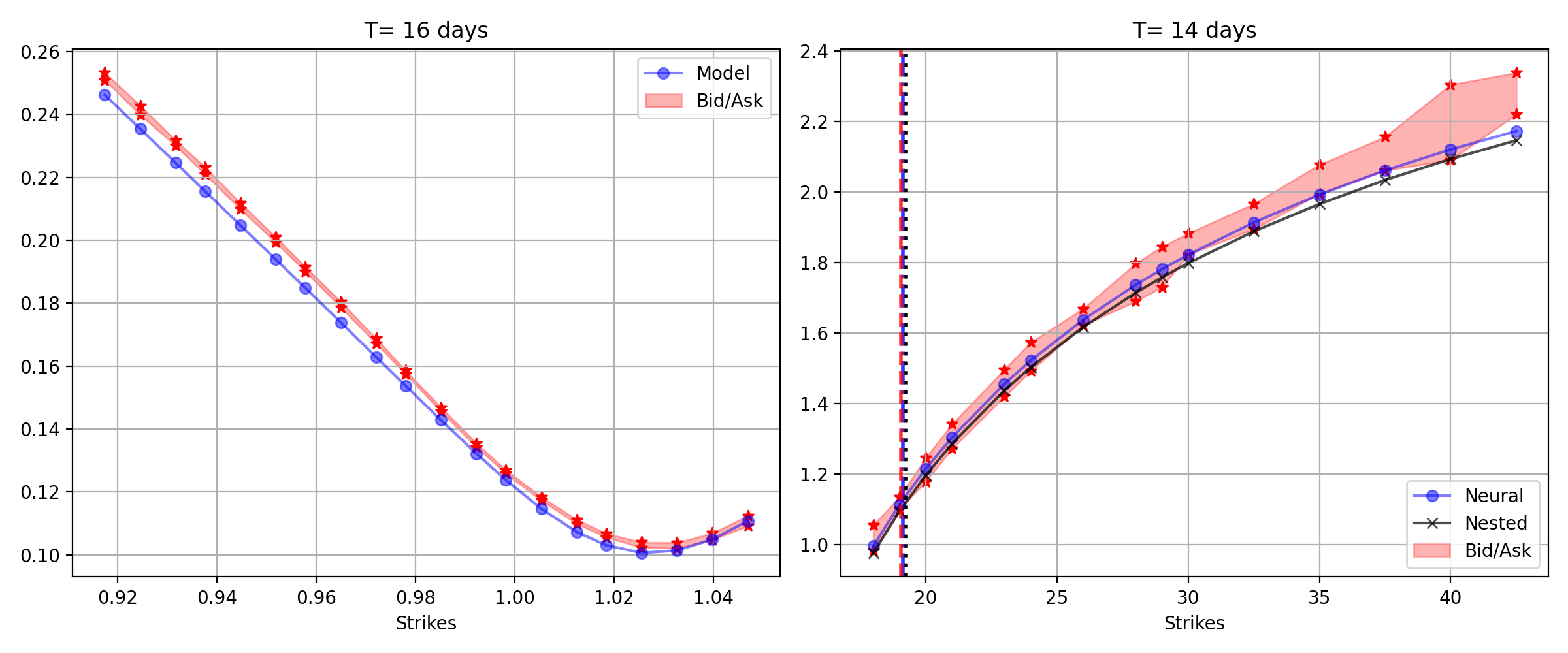}
         \includegraphics[trim={1.0cm 0.2cm 1.2cm 0.7cm},clip,scale=0.5]{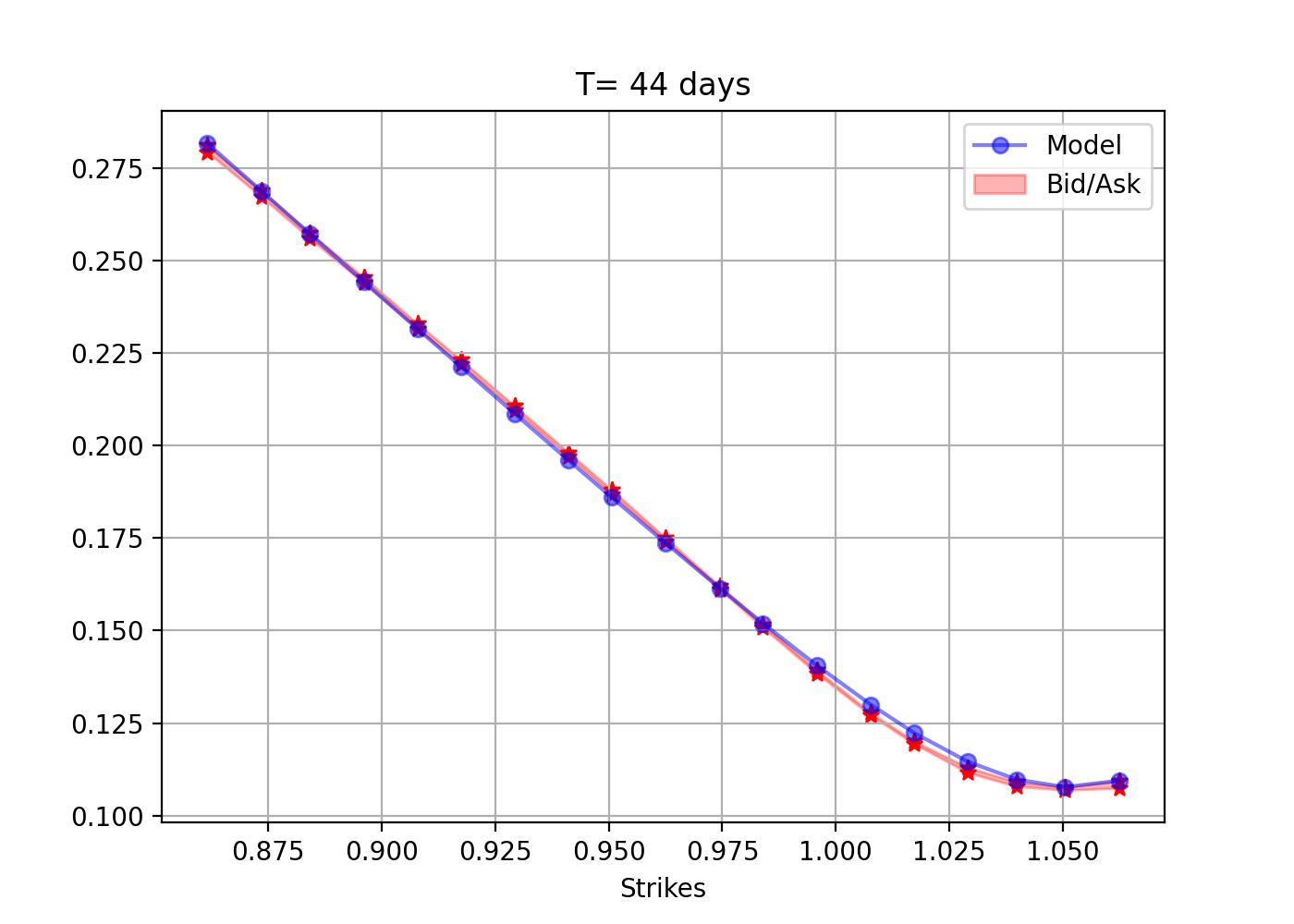}
         \caption{Joint calibration as of June 2, 2021. Comparison of calibrated and market SPX/VIX implied volatility smiles. { Top left and bottom: SPX smiles.} { Top right}: the market VIX future is displayed with dotted red lines while the model VIX future is indicated by the dotted blue line; in black (resp., dotted black) the VIX smile (resp., VIX future) computed with nested Monte Carlo. The initial values of the factors given the calibrated parameters in Table \ref{tab:joint11} are $R_{1,0}=0.2689$, $R_{1,1}=0.2375$, $R_{2,0}=0.0249$, $R_{2,1}=0.02491$.}
         \label{fig:joint2}
\end{figure}

\begin{table}[H]
    \begin{center} 
    \caption{{ Calibrated parameters of the 4FPDV model. Joint calibration to the SPX and VIX smiles as of June 2, 2021.}}
\label{tab:joint11}
\begin{tabular}{||c |c|c| c|c|c||}
 \hline
 $\lambda_{1,0}=44.42$ & $\lambda_{1,1}=33.19$ & $\theta_1=0.398$ &$\lambda_{2,0}= 4.311$ & $\lambda_{2,1}=3.254$    & $\theta_2=0.72$\\ 
 \hline
\end{tabular}

\begin{tabular}{|| c |c|c| c||} 
 \hline
  $\beta_{0}=0.0254$ & $\beta_{1}=-0.1602$ & $\beta_{2}=0.6922$ & $\beta_{1,2}=0.1639$\\ 
 \hline
\end{tabular}
\end{center}

\end{table}

\begin{figure}[H]
     \centering
         \includegraphics[width=\textwidth]{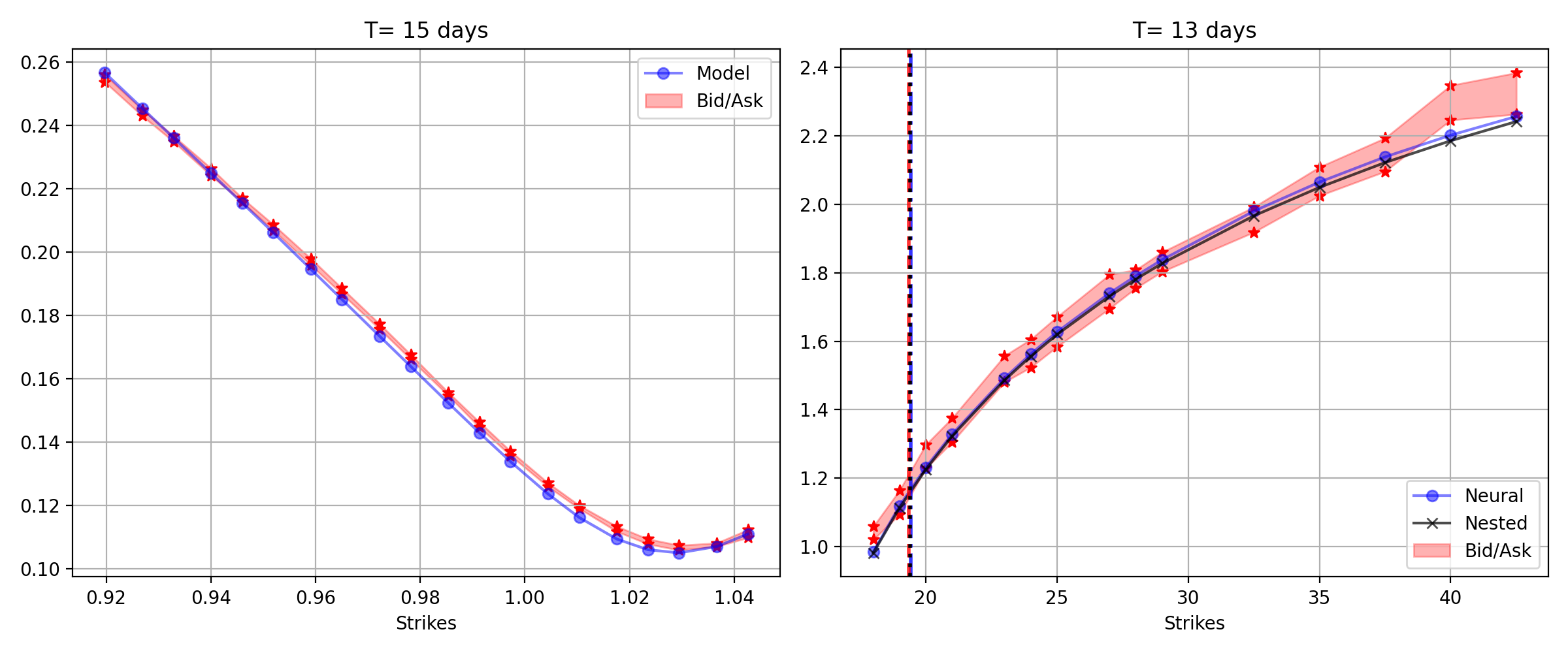}
         \includegraphics[trim={1.0cm 0.2cm 1.2cm 0.7cm},clip,scale=0.5]{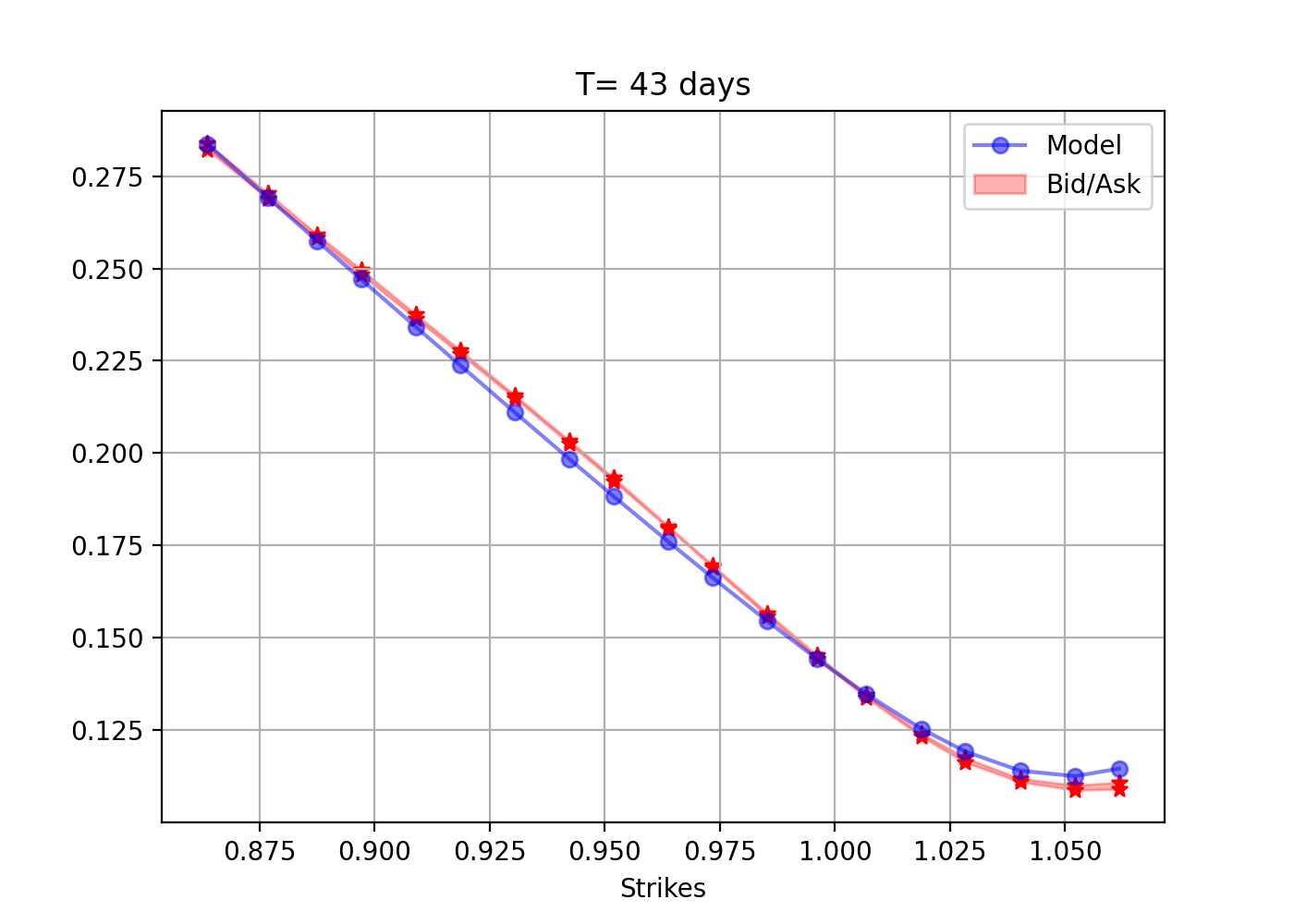}
         \caption{{ Joint calibration as of June 3, 2021}. See Figure \ref{fig:joint2}. { The initial values of the factors given the calibrated parameters in Table \ref{tab:joint1} are $R_{1,0}=0.0669$, $R_{1,1}=0.0916$, $R_{2,0}=0.02197$, $R_{2,1}=0.02725$.}}
         \label{fig:joint3}
\end{figure}

\begin{table}[H]
    \begin{center} 
    \caption{Calibrated parameters of the 4FPDV model. Joint calibration to the SPX and VIX smiles as of June 3, 2021.}
\label{tab:joint1}
\begin{tabular}{||c |c|c| c|c|c||}
 \hline
 $\lambda_{1,0}=42.78$ & $\lambda_{1,1}=31.51$ & $\theta_1=0.389$ &$\lambda_{2,0}=3.694$ & $\lambda_{2,1}=3.693$    & $\theta_2=0.698$\\ 
 \hline
\end{tabular}

\begin{tabular}{|| c |c|c| c||} 
 \hline
  $\beta_{0}=0.0264$ & $\beta_{1}=-0.1665$ & $\beta_{2}=0.6829$ & $\beta_{1,2}=0.1628$\\ 
 \hline
\end{tabular}
\end{center}
\end{table}
 
In Figures \ref{fig:joint2} and \ref{fig:joint3} we observe that we are able to calibrate jointly the first two monthly smiles of the SPX (top left and bottom) and the first monthly VIX future and VIX smile (top right) with a very good accuracy.
Note that for both calibration dates, $\lambda_{1,0}\approx \lambda_{1,1}$ and $\lambda_{2,0}\approx \lambda_{2,1}$, which means that one-exponential kernels seem enough for this fit. This is due to the fact that we only calibrate to short-dated options here. Recall that, by contrast, in Section \ref{sec:calib_SPX_surface}, we have seen that two-exponential kernels are needed to jointly fit short-dated and long-dated SPX implied volatilities.
For each iteration of the optimizer, we simulated $N_{MC}=2\cdot 10^{5}$ trajectories with discretization {step} $\Delta t=\frac{1}{504}$. The hyperparameters of the loss functions for SPX and VIX options are fixed to be $(\omega_{\SPX},\omega_{\VIX},\omega_{F})=(10,5,20)$. Like for the calibration to the SPX surface, we use the Py-BOBYQA optimizer and randomize the initial guess within the natural bounds of the parameters. The joint calibration takes around 8 minutes to complete.

In Figures \ref{fig:joint2} and \ref{fig:joint3}, we do not only report the VIX future and VIX implied volatilities in the calibrated model when we use our neural network approximation $\mathcal{N}\mathcal{N}^{\star}$ of the VIX. We also report those quantities when we use nested Monte Carlo paths to estimate the VIX. Figures \ref{fig:joint2} and \ref{fig:joint3} show that for the jointly calibrated parameters our neural network approximation of the VIX is accurate enough for trading purposes.

\section{On the stability of the calibrated parameters}\label{sec:stability}
In this section we address the stability of the calibrated parameters for the 4FPDV model. Stability of the calibrated parameters is desirable as, usually, models are periodically recalibrated for hedging purposes; parameter stability prevents oscillating hedge ratios and higher transaction costs.
A similar analysis has been carried out in \cite{CGMS:23}.

{ We split our analysis in two parts. In the first part (Section \ref{sec:stab1}), we calibrate OTM SPX monthly options every day of the fourth trading week of October 2023. In the second part (Section \ref{sec:stab2}), we compare the parameters found for the two joint calibration exercises in Section \ref{sec:joint}, i.e., for the two consecutive days June 2, 2021 and June 3, 2021. In both cases we display the corresponding kernels of the 4FPDV model as defined in \eqref{eq:kernel}.}

\subsection{Calibration to SPX options only}\label{sec:stab1}

\begin{table}
    \centering
\begin{tabular}{|c|}
  \hline
  \textbf{MAE} \\
  \hline
  October 23, 2023\qquad $2.23 \cdot 10^{-3}$ \\
  October 24, 2023\qquad $2.26 \cdot 10^{-3}$ \\
  October 25, 2023\qquad $2.12 \cdot 10^{-3}$\\
  October 26, 2023\qquad $2.63 \cdot 10^{-3}$ \\
  October 27, 2023\qquad $2.32 \cdot 10^{-3}$ \\
  \hline
\end{tabular}
    \caption{Mean absolute error of the SPX implied volatility fit. {The range of moneyness considered for each maturity $T>0$, is given by $[K_{\min},K_{\max}]$ where $K_{\min}=1-0.4\sqrt{T}$ and $K_{\max}=1+0.25\sqrt{T}$.}}
    \label{tab:MAE}
\end{table}

In this section, we calibrate OTM SPX monthly options every day of the fourth trading week of October 2023, that is, October 23 to October 27, 2023. We calibrate to monthly maturities up to one year. Figure \ref{fig:spx_surface2} shows the result of the calibration as of October 25, 2023. 
Table \ref{tab:MAE}, in which we report the  mean absolute error (MAE) of the implied volatility fit,
shows that the fits as of October 23, 24, 26, and 27, not reported here, achieve the same degree of accuracy.

{The values of the calibrated parameters on the five consecutive trading days are reported in Figure \ref{fig:combined_figure}. The $\beta$ parameters $(\beta_0,\beta_1,\beta_2,\beta_{1,2})$ (Figure \ref{fig:stability_parameters22}) are remarkably stable over time. The kernel parameters $(\lambda_{1,0},\lambda_{1,1},\theta_1,\lambda_{2,0},\lambda_{2,1},\theta_2)$ (Figure \ref{fig:stability_parameters11}) appear less stable, in particular the $K_2$ parameters. However, note that different values of $(\lambda_{1,0},\lambda_{1,1},\theta_1)$ (resp., $(\lambda_{2,0},\lambda_{2,1},\theta_2)$) produce similar two-exponential kernels. This is illustrated in Figure \ref{fig:kernels_onlyspx}, and is mostly observed for $K_1$.}

\begin{figure}
    \centering

        \begin{subfigure}{\textwidth}
        \centering
        \includegraphics[trim={0.7cm 0.4cm 0.7cm 0.9cm},clip,scale=0.45]{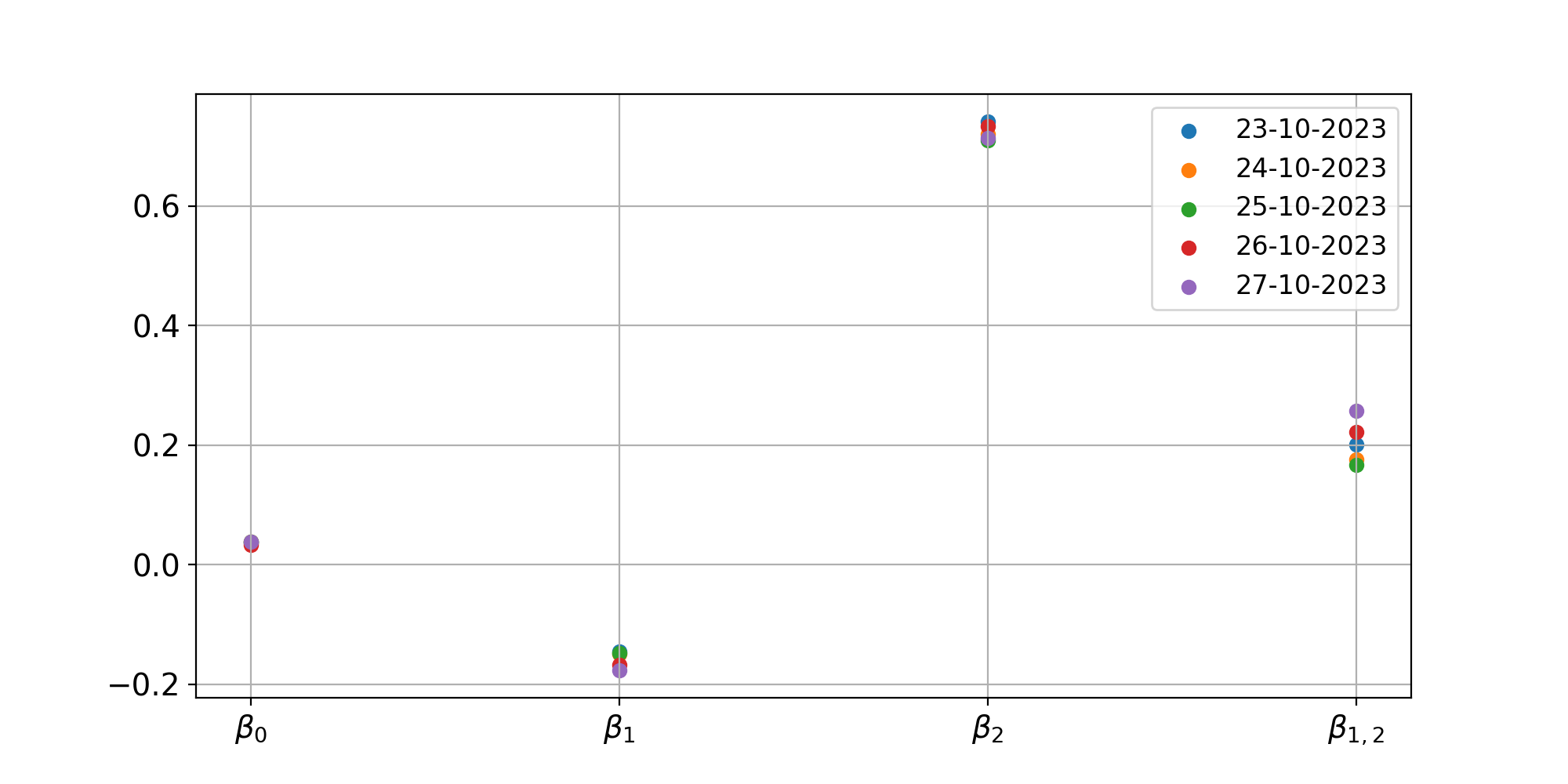}
        \caption{{Sensitivities $(\beta_0,\beta_1,\beta_2,\beta_{1,2})$ of the volatility for five consecutive calibration dates.}}
        \label{fig:stability_parameters22}
    \end{subfigure}

        \vspace{1em}
    
    \begin{subfigure}{\textwidth}
        \centering
        \includegraphics[trim={0.7cm 0.0cm 0.7cm 1.0cm},clip,width=\textwidth]{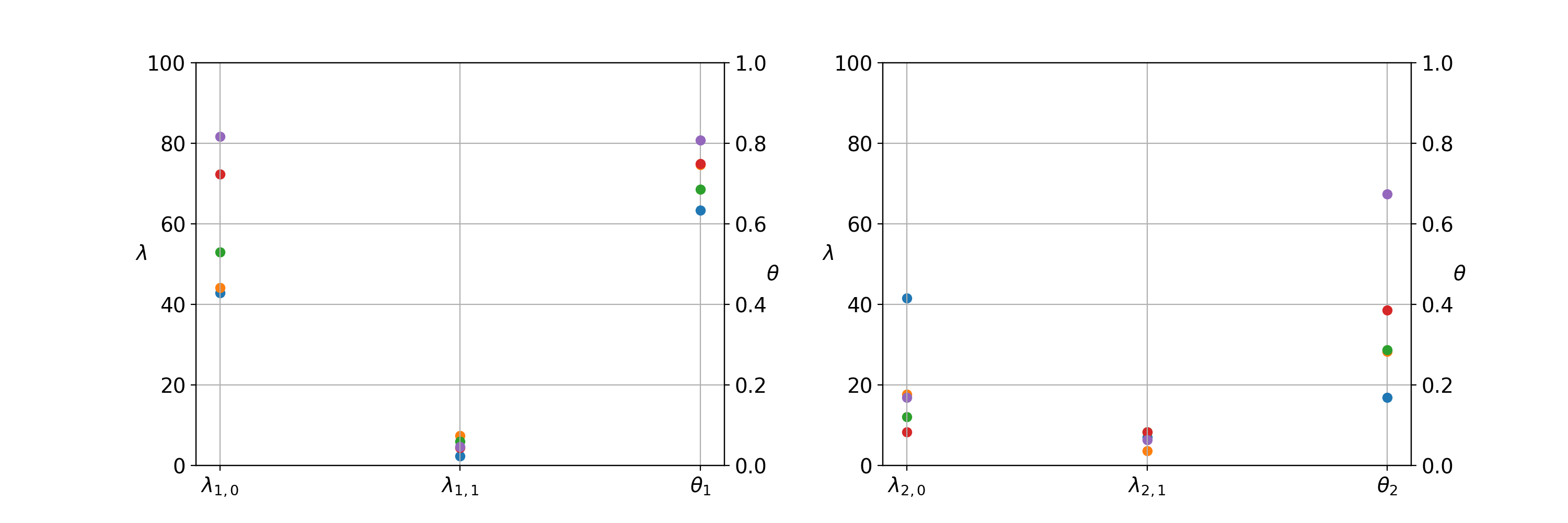}
        \caption{{Parameters of the kernels $K_1$ and $K_2$ for five consecutive calibration dates. L.h.s.: parameters $(\lambda_{1,0},\lambda_{1,1},\theta_1)$ of the kernel $K_{1}$. R.h.s.: parameters $(\lambda_{2,0},\lambda_{2,1},\theta_2)$ of the kernel $K_{2}$. For each plot, the $y$-axis on the l.h.s. refers to the natural scale of the $\lambda_{n,p}$s while the $y$-axis on the r.h.s. refers to the values of the $\theta_{p}$s.}}
        \label{fig:stability_parameters11}
    \end{subfigure}
    
    \caption{}
    \label{fig:combined_figure}
\end{figure}

     \begin{figure}
         \centering
         \includegraphics[trim={0.0cm 0.5cm 0.5cm 0.8cm},clip,width=\textwidth]{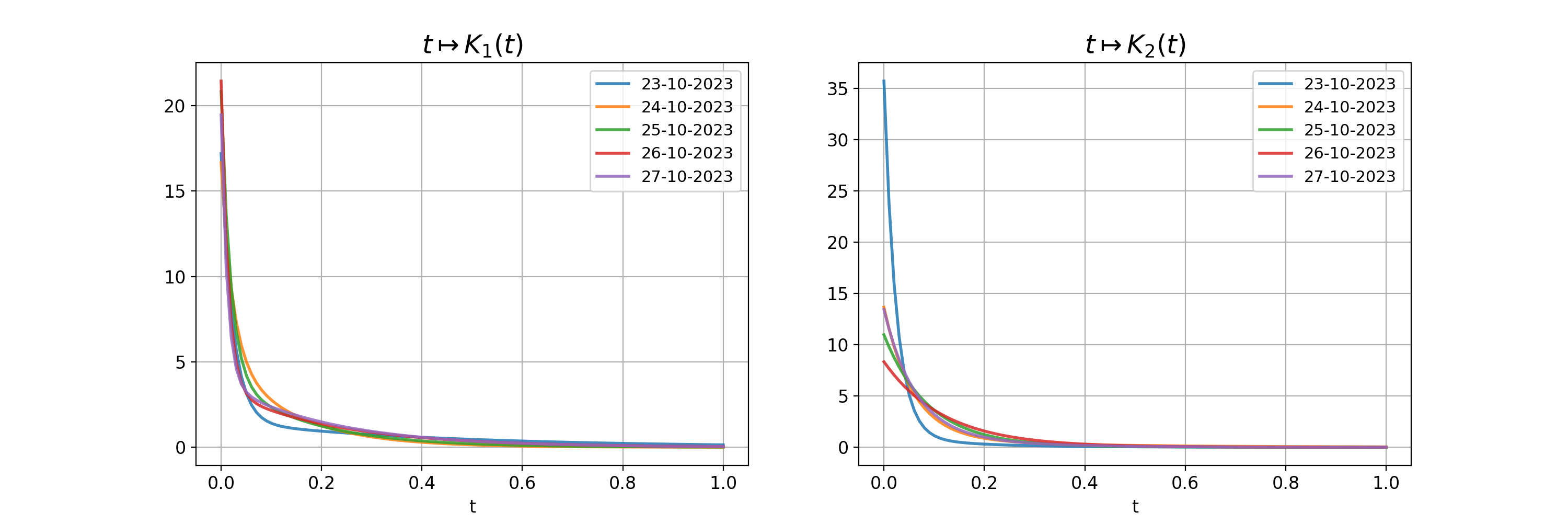}
         \caption{{Comparison of the kernels on five consecutive trading days. L.h.s.: $K_1$. R.h.s.: $K_2$}}
         \label{fig:kernels_onlyspx}
     \end{figure}

\subsection{Joint SPX/VIX calibration}\label{sec:stab2}

{In Figure \ref{fig:combined_figure2}, we plot the jointly calibrated parameters of Tables \ref{tab:joint11} (June 2, 2021) and \ref{tab:joint1} (June 3, 2021). We observe that the calibrated parameters are very stable from one calibration date to the next, including the kernel parameters.}
{In particular Figure \ref{fig:kernel_joint} illustrates that both kernels $K_1$ and $K_2$ do not change drastically from one calibration date to the next.}
\begin{figure}[H]
    \centering

        \begin{subfigure}{\textwidth}
        \centering
        \includegraphics[trim={0.7cm 0.4cm 0.7cm 0.9cm},clip,scale=0.45]{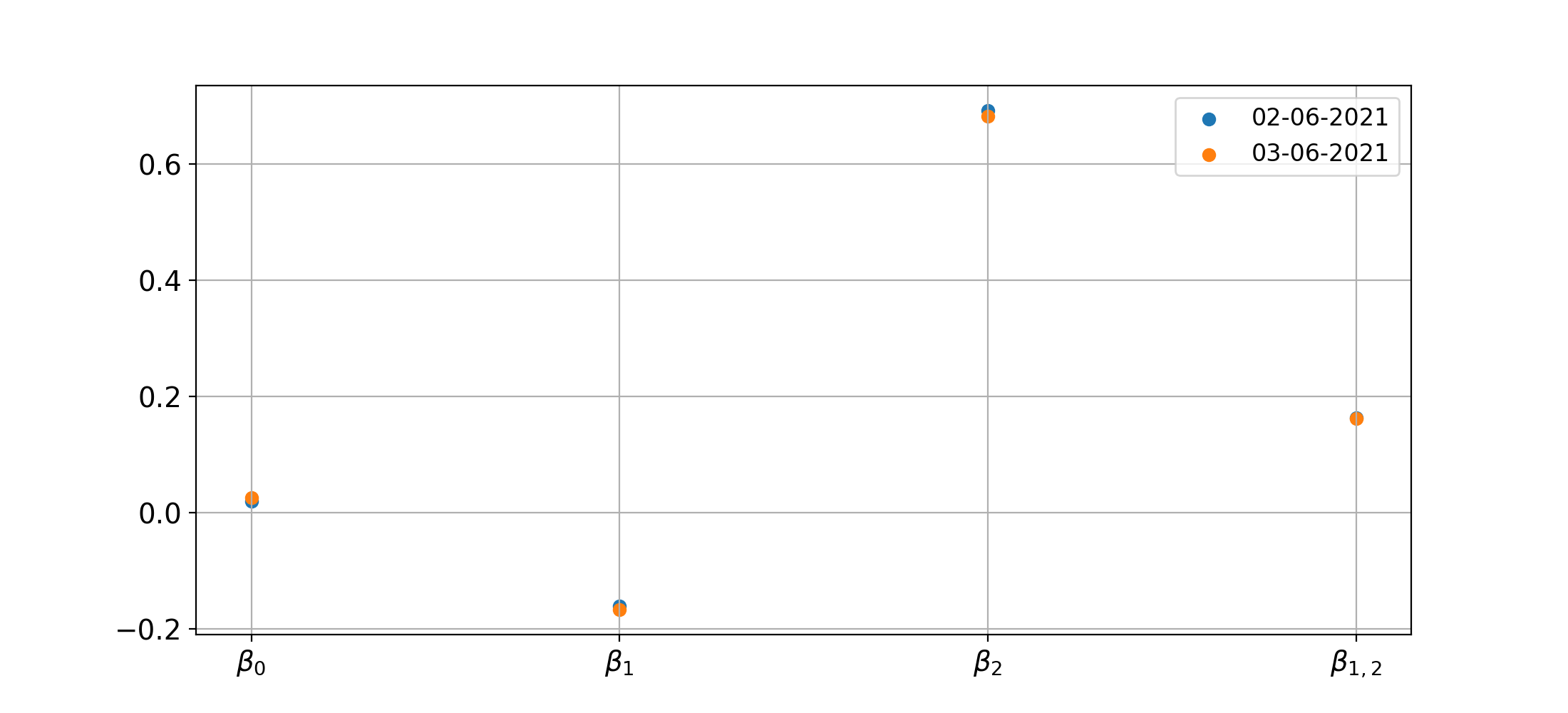}
        \caption{{Sensitivities $(\beta_0,\beta_1,\beta_2,\beta_{1,2})$ of the volatility for two consecutive calibration dates.}}
        \label{fig:stability_parameters22_joint}
    \end{subfigure}

        \vspace{1em} 
    
    \begin{subfigure}{\textwidth}
        \centering
        \includegraphics[trim={0.7cm 0.0cm 0.7cm 0.9cm},clip,width=\textwidth]{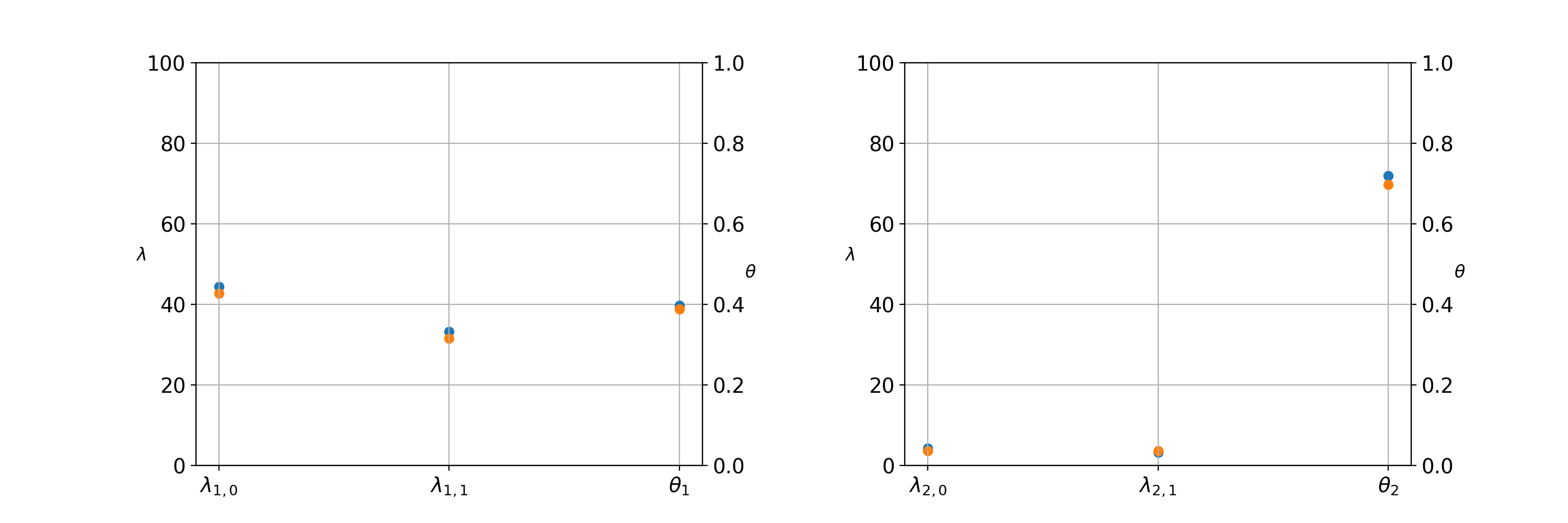}
        \caption{{Parameters of the kernels $K_1$ and $K_2$ for two consecutive calibration dates. L.h.s.: parameters $(\lambda_{1,0},\lambda_{1,1},\theta_1)$ of the kernel $K_{1}$. R.h.s.: parameters $(\lambda_{2,0},\lambda_{2,1},\theta_2)$ of the kernel $K_{2}$. For each plot, the $y$-axis on the l.h.s. refers to the natural scale of the $\lambda_{n,p}$s while the $y$-axis on the r.h.s. refers to the values of the $\theta_{p}$s.}}
        \label{fig:stability_parameters11_joint}
    \end{subfigure}
    
    \caption{}
    \label{fig:combined_figure2}
\end{figure}

\begin{figure}[H]
 \centering
 \includegraphics[trim={0.7cm 1.03cm 0.7cm 0.9cm},clip,scale=0.4]{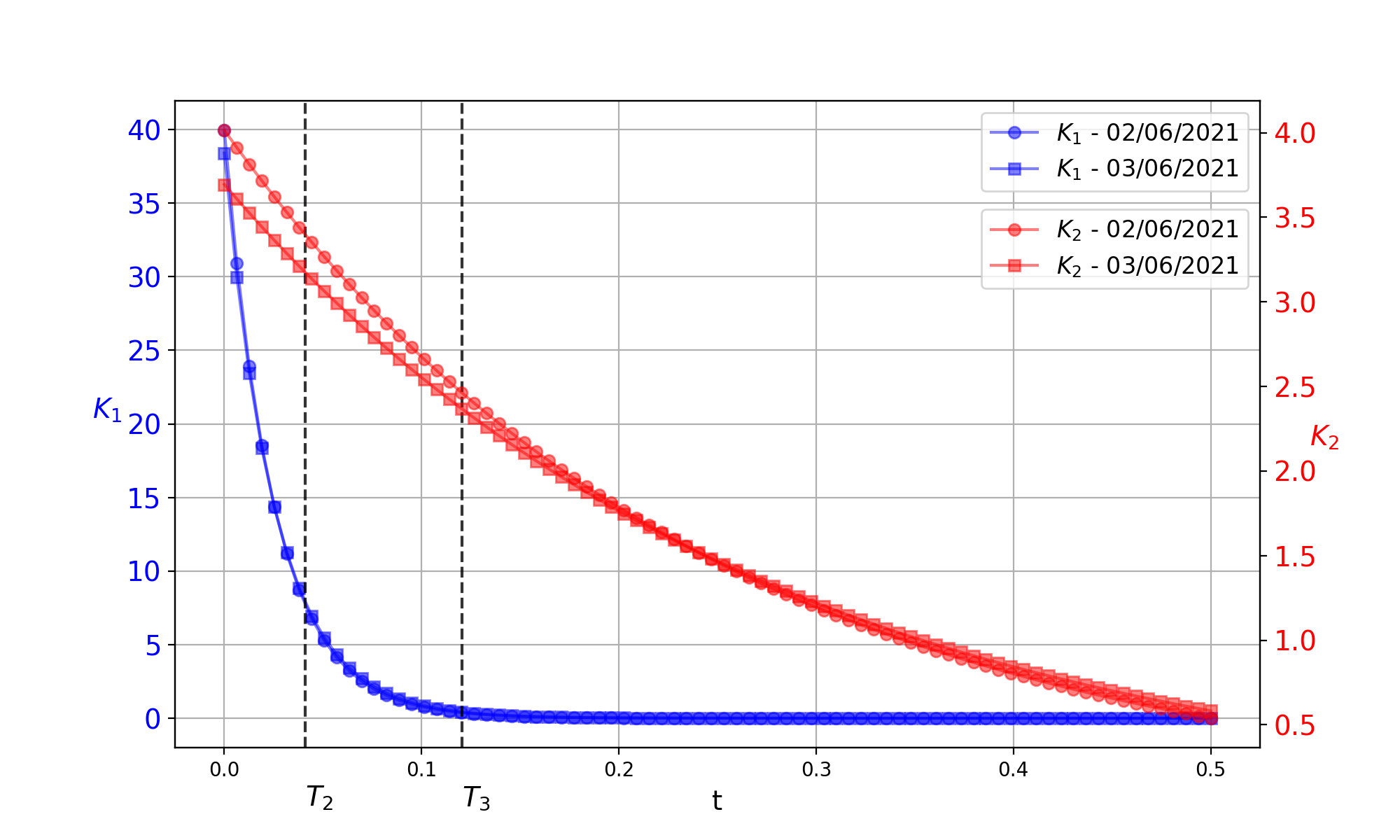}
 \caption{Comparison of the calibrated exponential kernels jointly to SPX and VIX. {The two vertical dotted lines denote the smallest maturity considered for calibration (13 days) and the largest one (44 days).}}
 \label{fig:kernel_joint}
\end{figure}

{\cb \section{Pricing light exotics} \label{sec:pricing_light_exotics}

One of the benefits of learning the VIX pathwise, and not only the prices at time 0 of VIX vanillas, is that in addition to a faster joint calibration discussed in the previous section, we are also able to quickly price derivatives involving the VIX. For instance, we consider the following light exotics, typically sold by banks to hedge funds, a VIX call with a lower barrier on the SPX (Payoff 1), and a VIX put with an upper barrier on the SPX (Payoff 2):
\begin{align*}
    &\text{Payoff 1}:\ \ (\text{VIX}_{T_1}-K)_+ {\bf 1}_{\{\inf_{t\in[0,T_1]} \frac{S_t}{S_0}\ge B\}}, \\
    &\text{Payoff 2}: \ \ (K-\text{VIX}_{T_1})_+ {\bf 1}_{\{\sup_{t\in[0,T_1]} \frac{S_t}{S_0}\le B\}}.
\end{align*}
The barriers are introduced to make the derivative much cheaper, as often the SPX has hit the low (resp., high) barrier when the VIX call (resp., put) is in the money. We can very quickly price these payoffs in the jointly calibrated 4FPDV model, using our fast pathwise approximation $\mathcal{N}\mathcal{N}^\star$ for the VIX. We use a discretization step of $\Delta t=\frac{1}{2520}$, $N_{MC}=2\times 10^6$ Monte Carlo samples, the model parameters reported in Table \ref{tab:joint1}, and $T_1=14$ days. Figures \ref{fig:exotics1} and \ref{fig:exotics2} show the impact of the strike and barrier levels.}

	\begin{figure}[H]
		\centering
		\begin{subfigure}{0.49\linewidth}
			\centering
			\includegraphics[trim={0.0cm 0.cm 1.0cm 0.0cm},scale=0.5]{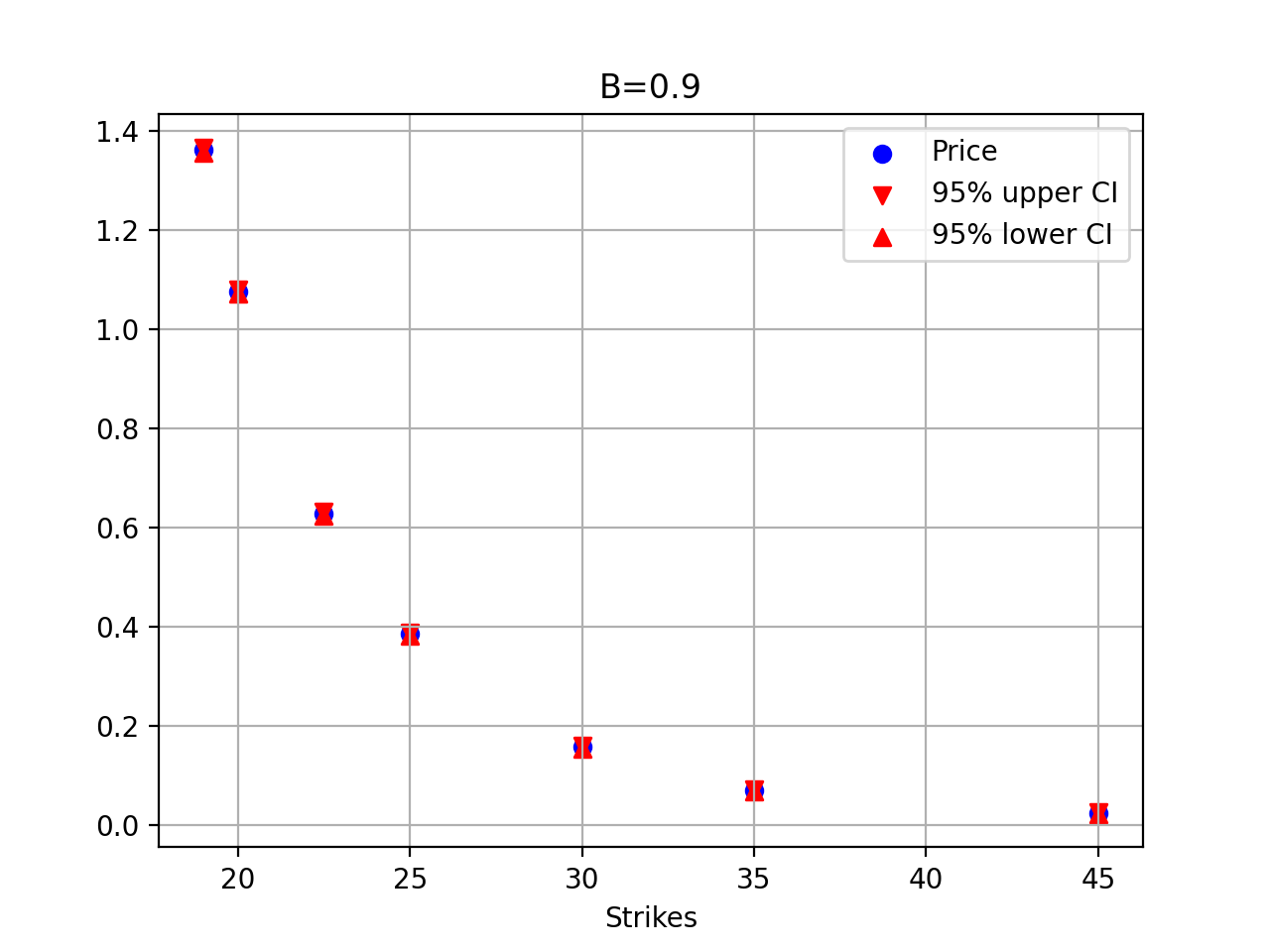}
		\end{subfigure}
		\hfill
		\begin{subfigure}{0.49\linewidth}
			\centering
			\includegraphics[trim={0.0cm 0.0cm 0.0cm 0.0cm},scale=0.5]{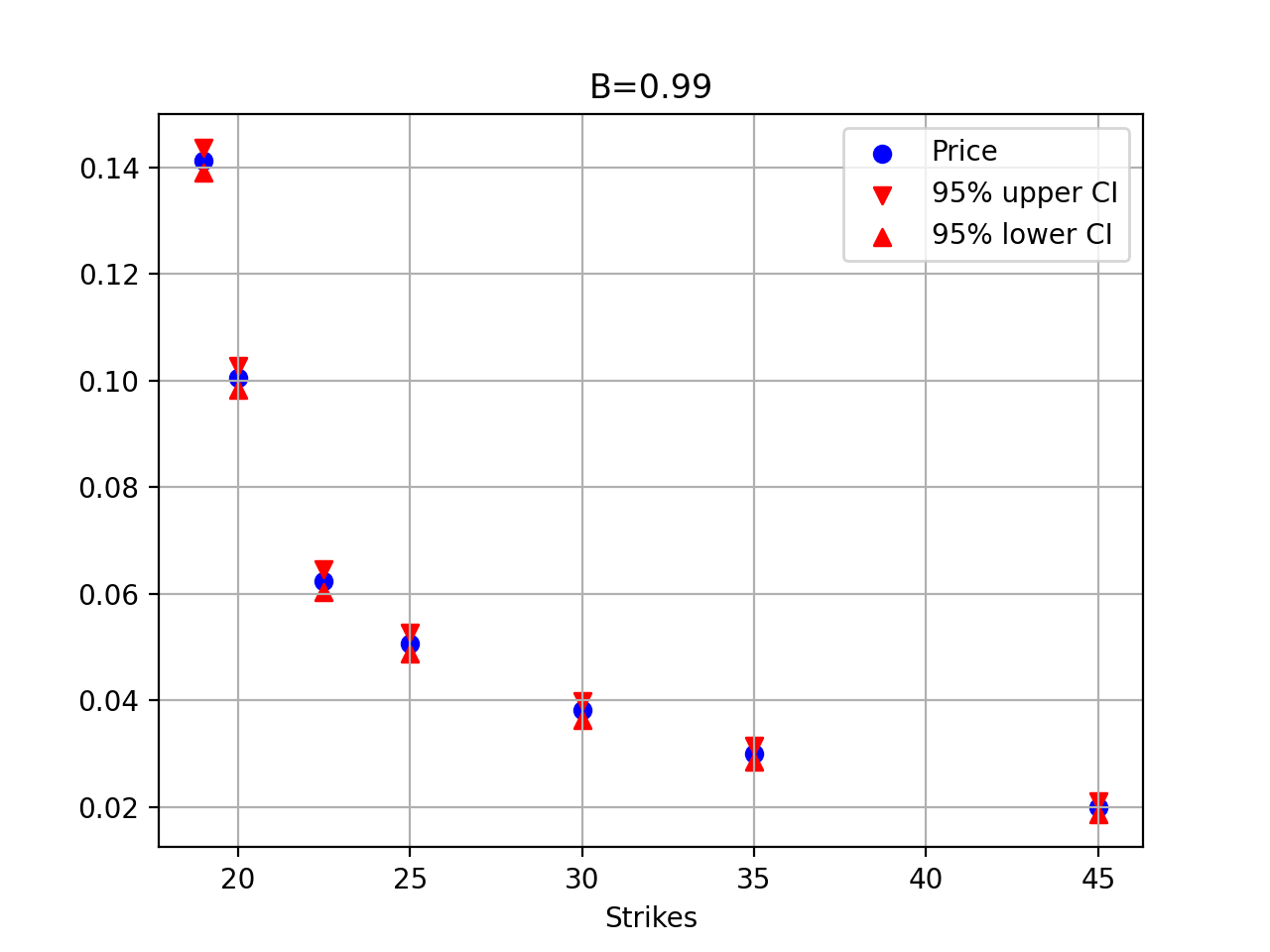}
		\end{subfigure}
		\caption{Payoff 1}
		\label{fig:exotics1}
	\end{figure}

\begin{figure}[H]
		\centering
		\begin{subfigure}{0.49\linewidth}
			\centering
			\includegraphics[trim={0.0cm 0.cm 1.0cm 0.0cm},scale=0.5]{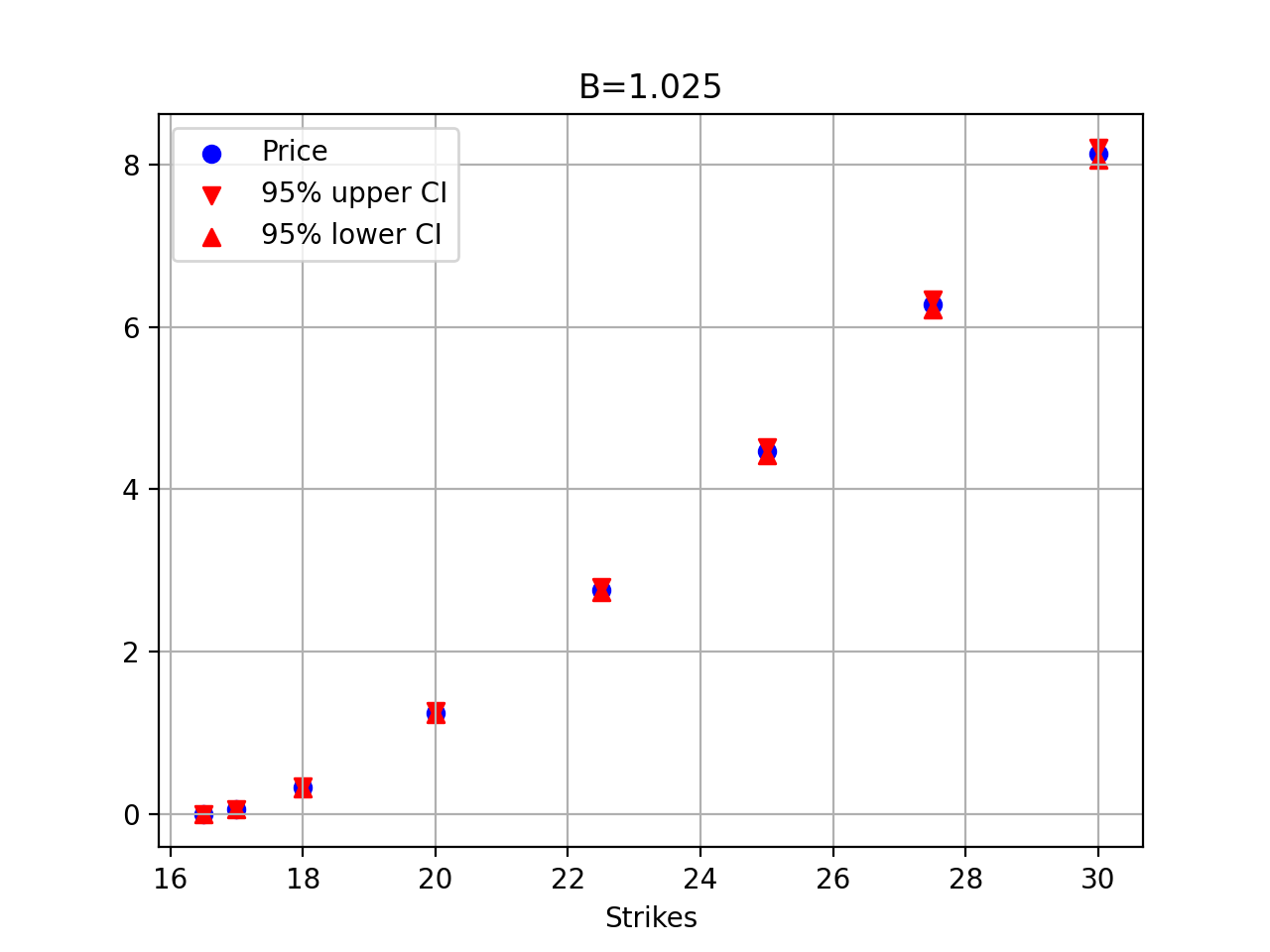}
		\end{subfigure}
		\begin{subfigure}{0.49\linewidth}
			\centering
			\includegraphics[trim={0.0cm 0.0cm 0.0cm 0.0cm},scale=0.5]{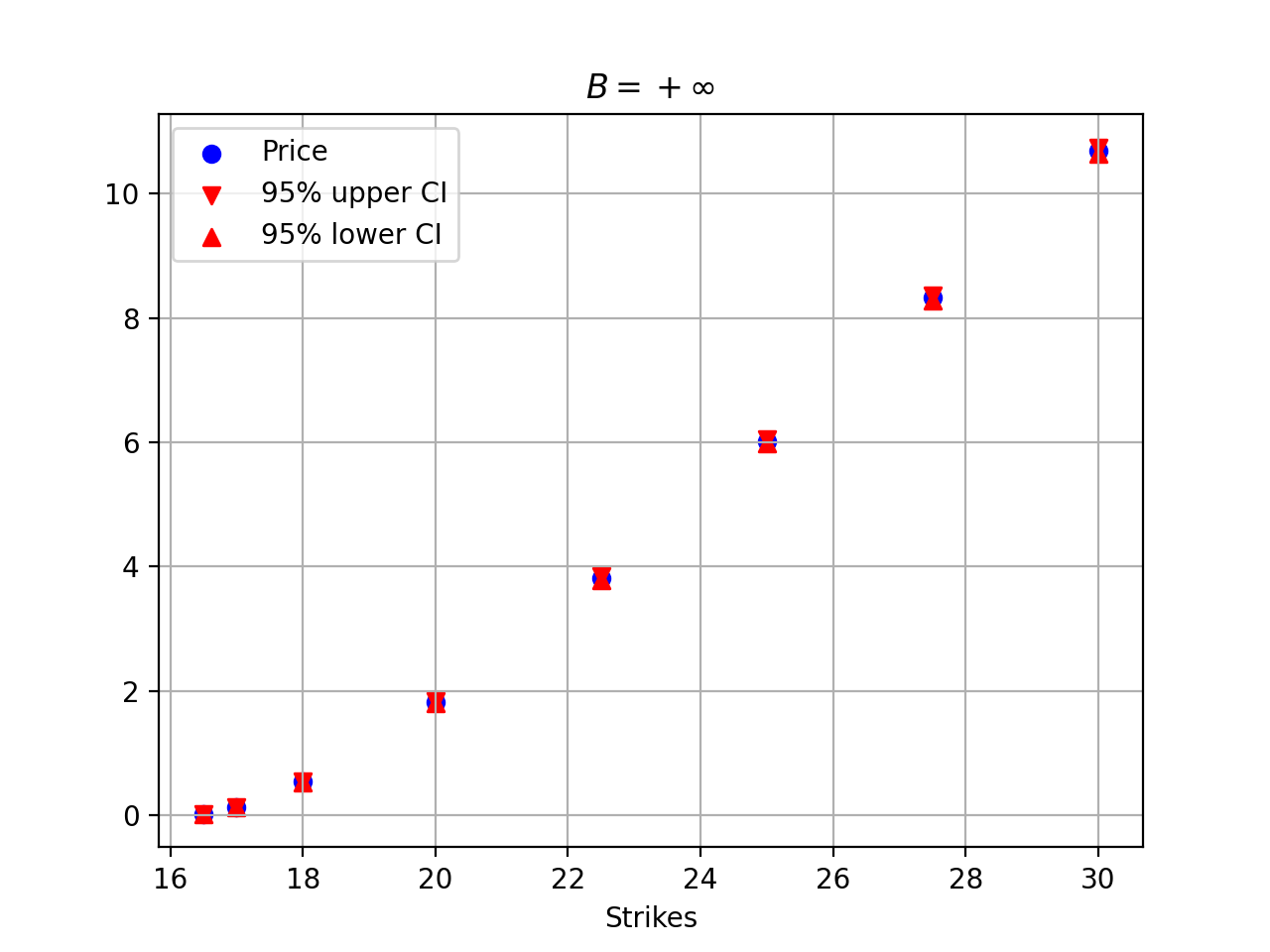}
		\end{subfigure}
		\caption{Payoff 2}
		\label{fig:exotics2}
\end{figure}

{\cb \section{Conclusion}

This paper investigates the VIX in path-dependent volatility (PDV) models, focusing on the 4FPDV model by \cite{GLJ:22}. It introduces a \emph{pathwise} neural network approximation of the VIX as a function of both the Markovian factors driving the volatility dynamics and the model parameters. The VIX is indeed not known in closed form in the model. For given model parameters, the dependence on the factors allows us to quickly sample VIX paths and price derivatives involving the VIX. The dependence of our neural net approximation on the model parameters allows us to quickly calibrate the model jointly to SPX and VIX options. Compared to the existing \emph{deep pricing} and \emph{deep calibration} techniques, our methodology allows us to have access to the VIX pathwise in the model, which would be lost if we learned directly the map to model option prices or model implied volatilities.

Also importantly, the paper shows, using a few random calibration dates, that the \emph{time-homogeneous}, low-parametric, Markovian 4FPDV model is able to fit the surface of SPX implied volatilities remarkably well.

Perspectives of future research include performing a large-scale empirical study to measure how well the 4FPDV model fits the SPX implied volatility surface over several years of calibration dates; jointly calibrating the model to more, longer maturities; measuring the impact of introducing a stochastic volatility component and/or discrete time; and investigating the tractability properties of affine or polynomial versions, directly written on the variance, of this PDV model, in which the VIX is known in analytical form.

\section*{Acknowledgements}
The authors thank the two anonymous referees for useful comments which helped improve this article. The first author would also like to thank Fabio Baschetti for fruitful discussions.}

\appendix

\section{Appendix: Additional tests of the neural network}\label{appendix:nn}

{ In this appendix, we report the error between the neural network predictions $\Ncal\Ncal^{\star}$ and the nested Monte Carlo estimator $\widehat{\VIX}$ for market-calibrated parameters, see Figures \ref{fig:test_Net} and \ref{fig:test_Net2}.}

\begin{figure}[H]
  \centering

  \begin{subfigure}{0.5\textwidth}
    \includegraphics[width=\linewidth]{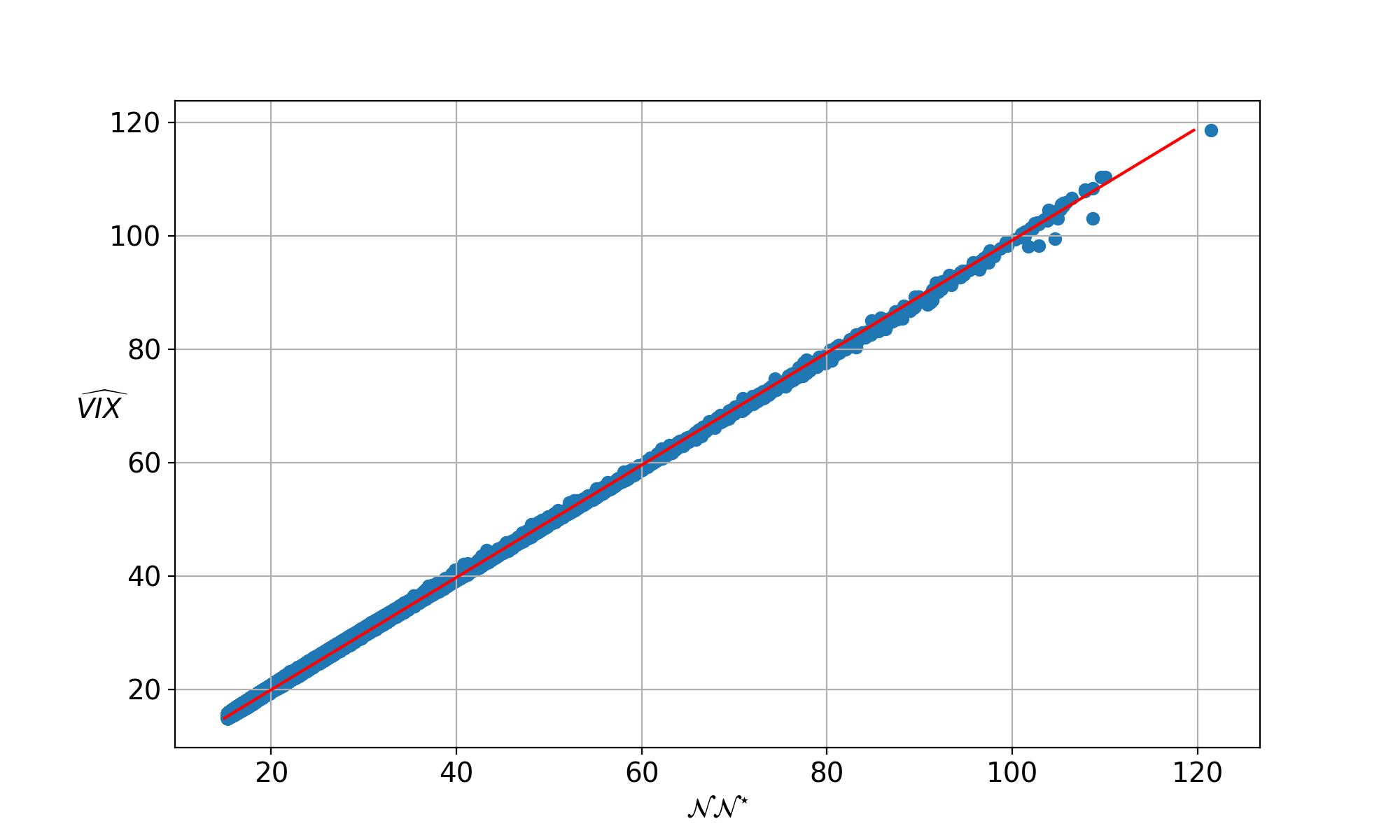}
    \caption{}
  \end{subfigure}
  \hfill
  \begin{subfigure}{0.46\textwidth}
    \includegraphics[width=\linewidth]{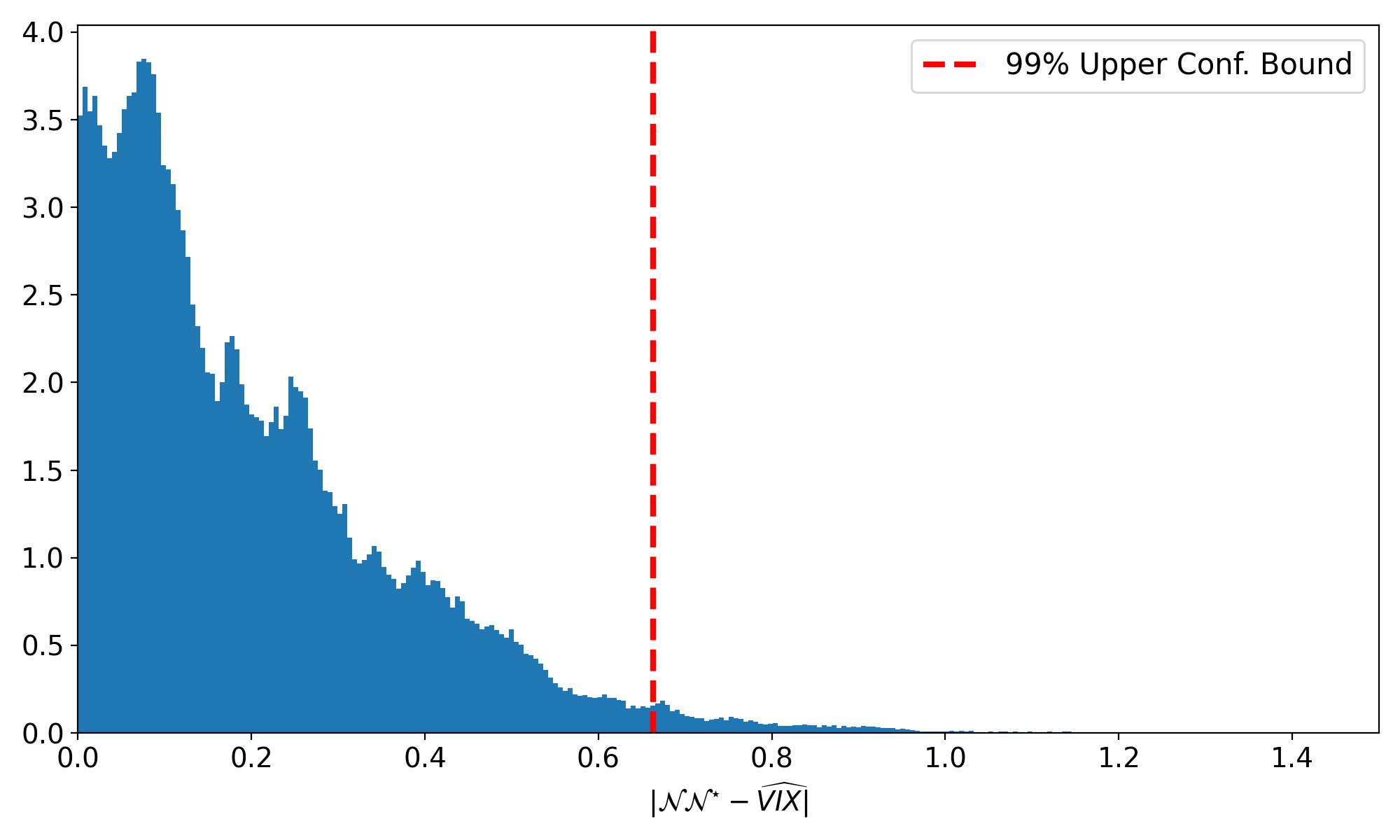}
    \caption{}
  \end{subfigure}

  \caption{{ Model parameters are jointly calibrated on June 2, 2021. Left: comparison between $N_{MC}=3\cdot 10^5$ realizations of the nested Monte Carlo estimator $\widehat{\VIX}$ and the neural approximation $\Ncal\Ncal^{\star}$. Here the $\widehat{\VIX}$ is computed with $\Delta t=\frac{1}{2520}$, $3\cdot 10^5$ outer paths and $10^{4}$ nested paths. Right: histogram of the absolute error between the neural network predictions $\Ncal\Ncal^{\star}$ and $\widehat{\VIX}$. The average absolute error is 0.202$\%$.}}
  \label{fig:test_Net}
\end{figure}

\begin{figure}[H]
  \centering

  \begin{subfigure}{0.5\textwidth}
    \includegraphics[width=\linewidth]{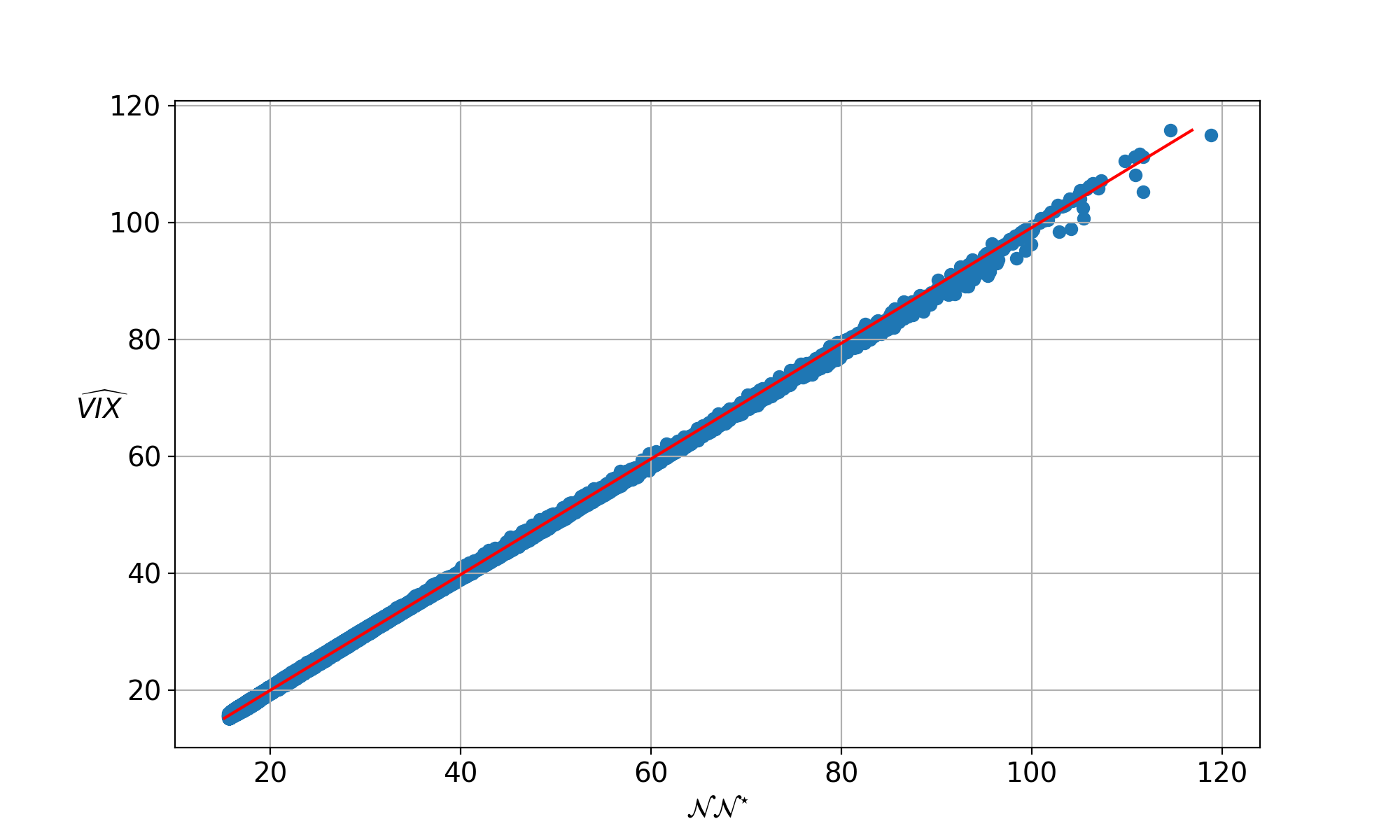}
    \caption{}
  \end{subfigure}
  \hfill
  \begin{subfigure}{0.46\textwidth}
    \includegraphics[width=\linewidth]{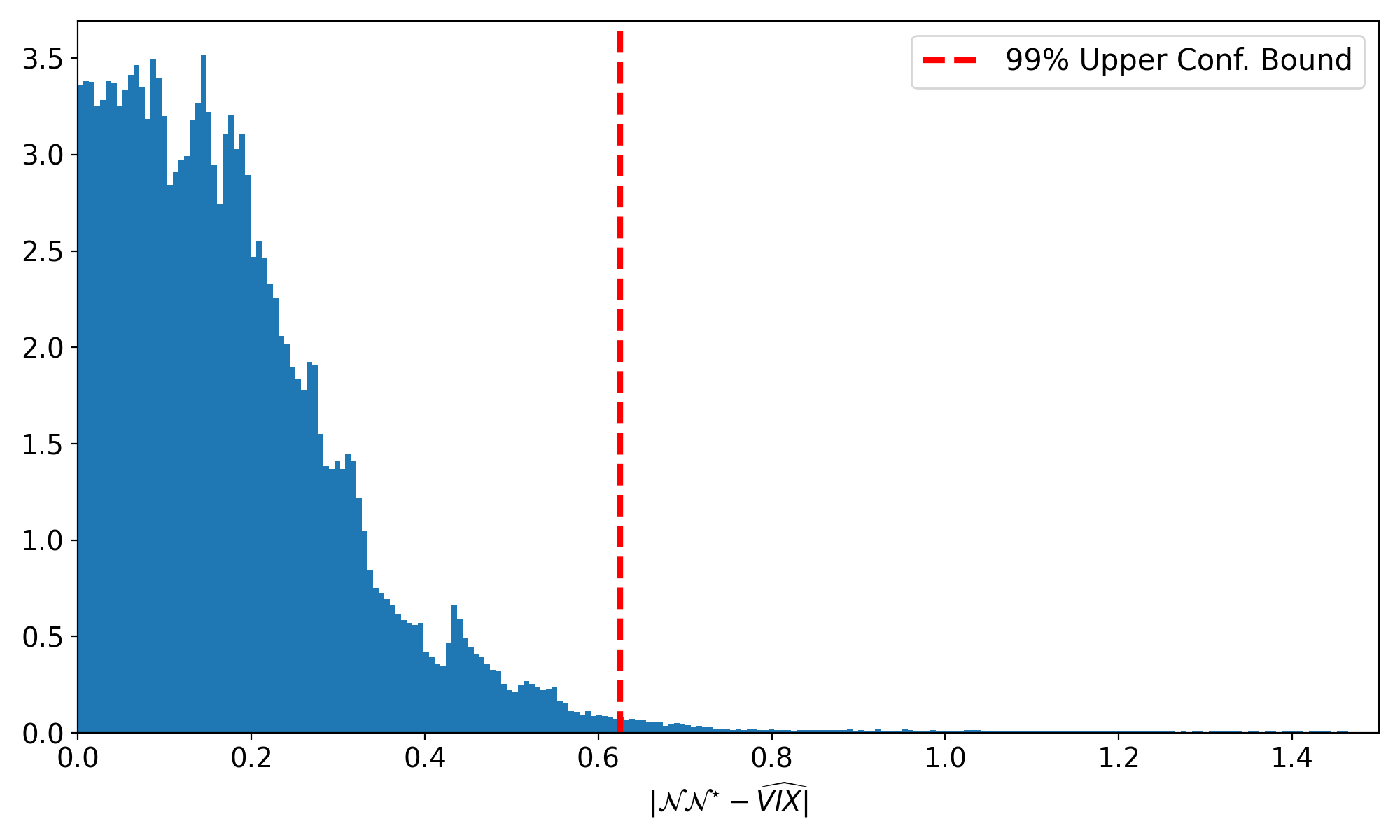}
    \caption{}
  \end{subfigure}

  \caption{{Same as Figure \ref{fig:test_Net}, with model parameters jointly calibrated on June 3, 2021. The average absolute error between the neural network predictions $\Ncal\Ncal^{\star}$ and $\widehat{\VIX}$ is 0.185$\%$.}}
  \label{fig:test_Net2}
\end{figure}

\section*{Disclosure of interest}
The authors declare that there are no relevant financial or non-financial competing interests to report.

\appendix

\end{document}